\documentclass[a4paper]{article}
\usepackage{a4wide}
\usepackage{graphicx,xcolor}
\graphicspath{{./images/}}
\usepackage{caption}
\usepackage{subcaption}
\usepackage{empheq}
\usepackage{amsthm}
\usepackage{amssymb, latexsym, mathrsfs}
\usepackage{dsfont}
\usepackage{enumerate}
\usepackage{algorithm}
\usepackage{color}
\usepackage{transparent}
\usepackage{url}
\usepackage[affil-it]{authblk}
\usepackage{booktabs}

\usepackage[utf8]{inputenc}
\usepackage[T1]{fontenc}
\usepackage{lmodern}

\usepackage[american,cuteinductors,smartlabels]{circuitikz}
\usepackage{tikz}
\usepackage{schemabloc}
\usetikzlibrary{shapes,arrows}

\theoremstyle{plain}
\newtheorem{theorem}{Theorem}
\newtheorem{assum}{Assumption}
\newtheorem{prbl}{Problem}

\newtheorem{rmk}{Remark}
\newtheorem{prp}{Proposition}

\def\Rset{\mathbb{R}}

\newcommand{\AAA}{{\mathcal A}}

\newcommand{\DD}{{\mathcal D}}

\newcommand{\FF}{{\mathcal F}}
\newcommand{\GG}{{\mathcal G}}

\newcommand{\KK}{{\mathcal K}}
\newcommand{\LL}{{\mathcal L}}

\newcommand{\OO}{{\mathcal O}}

\newcommand{\QQ}{{\mathcal Q}}

\newcommand{\YY}{{\mathcal Y}}

\newcommand{\hAAA}{\hat{\mathcal{A}}}

\newcommand{\tQQ}{\tilde{\mathcal{Q}}}

\newcommand{\diag}{\mbox{diag}}

\newcommand{\matrice}[2]{\left[\ba{#1} #2 \ea\right]}

\newcommand{\ba}[1]{\begin{array}{#1}}
\newcommand{\ea}{\end{array}}

\begin{document}

  	\title{A scalable line-independent design algorithm for voltage and frequency control in AC islanded microgrids}
  		
  		\author[1]{Michele Tucci%
  		   \thanks{Electronic address:
  			\texttt{michele.tucci02@universitadipavia.it}}}
  		\author[2]{Giancarlo Ferrari-Trecate%
       		\thanks{Electronic address: \texttt{giancarlo.ferraritrecate@epfl.ch}; Corresponding author}}
		\affil[1]{Dipartimento di Ingegneria Industriale e
			dell'Informazione\\Universit\`a degli Studi di Pavia, Italy}
  		\affil[3]{Automatic Control Laboratory, \'Ecole Polytechnique F\'ed\'erale de Lausanne (EPFL), Switzerland}
  		     \date{\textbf{Technical Report}\\ October, 2018}

  		\maketitle
  		
  		\begin{abstract}
  		We propose a decentralized control synthesis procedure for stabilizing
  		voltage and frequency in AC Islanded microGrids (ImGs) composed of Distributed
  		Generation Units (DGUs) and loads interconnected through power lines. The presented approach
  		enables Plug-and-Play (PnP) operations, meaning that DGUs
  		can be added or removed without compromising the overall ImG
  		stability. The main feature of our approach is that the proposed design algorithm is line-independent. This implies that (i) the synthesis of
  		each local controller requires only the parameters of the corresponding DGU and not the model of power
  		lines connecting neighboring DGUs, and (ii) whenever a new DGU
  		is plugged in, DGUs physically coupled with it do not have to retune
  		their regulators because of the new power line connected to them. Moreover, we formally prove that stabilizing local controllers can be always computed, independently of the electrical parameters. Theoretical results are validated by
  		simulating in PSCAD the behavior of a 10-DGUs ImG.
  		
  	\end{abstract}

\newpage

\section{Introduction}
\label{sec:intro}
Voltage and frequency stability is a central problem in low-voltage AC
Islanded microGrids (ImGs) and, in the recent years, it
has received great attention within the control and the power
electronics communities \cite{Guerrero2013}. In
absence of a connection to the main grid (which acts as
an infinite power source and as a master clock for the ImG frequency),
voltage and frequency must be governed by the local
controllers of the Voltage Source Converters (VSCs) interfacing power
sources with the ImG. Each controlled VSC, together with its power
supply, forms a Distributed Generation Unit (DGU) connected to loads
and other DGUs through power lines. Voltage and frequency control can be then formulated as the problem of designing decentralized regulators guaranteeing collective ImG stability in spite of the electrical coupling between DGUs. Approaches to the decentralized control of ImGs can be divided
into two main classes. The first one embraces droop controllers \cite{Guerrero2013}, which mimic standard regulators for power
networks with inertial generators. Droop controllers admit a simple
implementation and do not require synchronized clocks for the
computation of the control signals. However, they can generate
frequency deviations, whose compensation calls for the use of secondary
distributed controllers \cite{Guerrero2013}. Stability
properties of droop-controlled microgrids have been analyzed in
\cite{schiffer2014conditions,Simpson-Porco2013} under simplified models of the DGU dynamics. The
second class of controllers comprises solutions based on
approaches developed within the field of decentralized control \cite{7172122,Etemadi2012a,Babazadeh2013,7062912,Riverso_TSG}. If,
on the one hand, they require controller clocks to be synchronized with
sufficient precision (through, e.g., GPS or communication networks
\cite{IEEE2017}), on the other hand, they enjoy built-in stability and robustness properties.

\indent Control design algorithms for ImGs can be also categorized according to
their level of modularity. Notably, the necessity to address the growing demand for flexible and resizable ImG structures \cite{kumagai2014rise} (where DGUs and loads can enter/leave over time), calls for a control architecture that can be easily updated when the ImG topology changes. Approaches of
this kind have been often termed Plug-and-Play (PnP)
\cite{7511679,tucci2015decentralized,Riverso_TSG,7039383,7040312}. In particular, in \cite{Riverso_TSG}, PnP means that (i) the
computation of a local controller for a DGU requires only the DGU model and the models of power
lines connected to it, (ii) the design of a local controller preserving
voltage stability in the whole ImG amounts to an optimization
problem. Note that, in view of (i), no global model of the ImG is
needed in control design. Moreover, the plug-in of a DGU requires to
update controllers of neighboring DGUs, at most. In addition, from
(ii), the plug-in of a DGU can be automatically denied if stabilizing
controllers for the DGU and its neighbors cannot be found. Approaches
to the design of distributed regulators with similar PnP features for
large-scale systems have been proposed in \cite{lucia2015contract,bodenburg2015plug,bendtsen2013plug,Riverso2013c}.

In this paper, we develop a PnP control design method that, differently
from \cite{Riverso_TSG}, does not require knowledge of power lines whose parameters are often uncertain; the only global quantity used in the synthesis algorithm is a scalar parameter. We do not assume either to know bounds on electrical coupling parameters, as done in \cite{7511679}. These simplifications are desirable for several reasons. First, the addition/removal of a DGU does not require to update any existing controller in the ImG. Indeed, plugging in/out operations do add/remove lines connected to neighboring DGUs, but DGU controllers are line-independent. Another feature of the proposed control design procedure is that DGUs with identical electrical parameters can be equipped with the same regulator, which can be computed off-line only once. Therefore, for ImGs using a limited set of VSC models, no control synthesis is required at the plug-in/-out time of a new DGU. In addition, while the PnP design in \cite{Riverso_TSG} is dependent on a global tuning parameter, which must be sufficiently small for ensuring collective ImG stability, here this constraint is removed. Indeed, we propose a different proof of voltage and frequency stability. Notably, we first exploit the fact that DGU interactions can be represented by the admittance matrix of the electric graph (which has a Laplacian structure) for guaranteeing the decrease of a separable Lyapunov function along state trajectories. Then, we complement this result with the application of the LaSalle invariance principle.

Another important feature of the control design procedure is that local stabilizing regulators always exist, independently of the electrical parameters of the DGUs, and they can be computed by solving Linear Matrix Inequality (LMI)  problems.

The approach taken in this paper share similarities with the one
in \cite{Tucci2016independent}, where DC mGs
are considered and a line-independent variant of the PnP design
algorithm in \cite{tucci2015decentralized} has been proposed. There
are, however, fundamental differences. First, in the AC case, one must handle
three-phase balanced signals or, in an equivalent way, their $dq$
representation \cite{schiffer2016survey}. This makes stability analysis more complex and, differently from 
\cite{Tucci2016independent}, our rationale hinges on a suitable 
reparametrization of local controllers and Lyapunov functions.
Second, compared to \cite{Tucci2016independent} the LMIs associated to local control design involve a different set of optimization variables and are guaranteed to be feasible.

The paper is structured as follows. The ImG model and the local control architecture is introduced in Section \ref{sec:ImG_model}. In Section \ref{sec:design}, we (i) derive local design conditions implying asymptotic stability of the whole ImG, (ii) present the algorithm for synthesizing line-independent PnP controllers along with the main stability theorem, and (iii) formulate control design through LMIs problems and discuss plug-in and -out operations. 
Simulation results using a 10-DGU ImG with linear and nonlinear loads are described in Section \ref{sec:simulations}. 
We also consider the case of shifts in the individual DGU clocks, showing that they can affect performance of the ImG but not stability.
A preliminary version of this work will be presented at the 20th IFAC World Congress. The design procedure in the conference version, however, is not guaranteed to be always feasible and it includes an additional nonlinear constraints. Moreover, differently from the conference version,  the present paper includes the proofs of Propositions 1-2, as well as additional results (i.e. Propositions 3 and 4) which are exploited in the proof of Theorem 1 (omitted in the conference paper). Finally, the present work includes more comprehensive simulation results.

\noindent \textbf{Notation and basic definitions.} The
$N$ by $N$ identity and null matrices are denoted, respectively, with $I_{N}$
and $\mathbf{0}_N$, while we use $\mathbf{0}$ to
indicate a null matrix of appropriate size. Let
$A\in\mathbb{R}^{n\times m}$ be a
matrix inducing the linear map $A:\mathbb{R}^m\rightarrow \mathbb{R}^n$. The \textit{image} and the \textit{nullspace} (or
\textit{kernel}) of $A$ are indicated with $\text{Im}(A)$ and
$\text{Ker}(A)$, respectively. Consider the subspace $B\subseteq\mathbb{R}^m$: with $A \big|_{B}$ we denote the
restriction of the map $A$ to the domain $B$. The
symbol $\oplus$ refers to the sum of subspaces that are orthogonal (also called \textit{orthogonal direct sum}).

\section{Microgrid model}
\label{sec:ImG_model}
\label{sec:el_model}
In this Section, we present the electrical model of the ImG we
used. We assume three-phase electrical signals without zero-sequence
components and balanced network parameters (see \cite{schiffer2016survey} for basic
definitions about AC three-phase signals).

The single-phase equivalent scheme of DGU $i$ is shown in the left 
dashed frame of Figure \ref{fig:ctrl_part}. In particular, we have a DC
voltage source for modeling a generic renewable resource and
a VSC is controlled in order to supply a
local load, connected to the Point of Common Coupling (PCC) through an
$RLC$ filter. 
We highlight that, as shown in \cite{dorfler2013kron,TucciFloriduz_Kron}, also more general
interconnections of loads and DGUs can be mapped into this setting by means of a network
reduction method known as Kron reduction.
We further assume that
loads $I_{Li}$ are unknown and act as current
disturbances \cite{Babazadeh2013}. More in general, we can consider an ImG composed of $N$ DGUs and
define the set $\mathcal{D}=\{1,\dots,N\}$. Then, we (i) call two DGUs
neighbors if there is a
power line connecting their PCCs, (ii) denote with
$\mathcal{N}_i\subset\mathcal{D}$ the subset of neighbors of DGU
$i$, and (iii) observe that the neighboring relation is
symmetric (i.e. $j\in\mathcal{N}_i$ implies
$i\in\mathcal{N}_j$). We indicate with $\mathcal{G}_{el}$
the undirected electric graph induced by the
neighboring relation over the node set $\mathcal{D}$.

Let $\omega_0$ be the reference network frequency. As in \cite{Riverso_TSG,TucciFloriduz_Kron}, the $i$-th DGU model in $dq$ reference frame
(rotating with speed $\omega_0$) is

\begin{equation}
\label{DGUeqx}
\text{DGU}~i:\hspace{-4mm}\quad\left\lbrace
\begin{aligned}
\frac{\mathrm d}{\mathrm dt} V_i^{dq}&=-\mathrm i \omega_0
V_i^{dq}+
\frac{I_{ti}^{dq}}{C_{ti}}-\frac{I_{Li}^{dq}}{C_{ti}} +\sum_{j\in\mathcal{N}_i} \frac{1}{C_{ti}}\left(\frac{V_j^{dq}}{R_{ij}+\mathrm i  \omega_0 L_{ij}}-\frac{V_i^{dq}}{R_{ij}+\mathrm i  \omega_0 L_{ij}}\right) \\
\frac{\mathrm d}{\mathrm dt}I_i^{dq} &=-\left( \frac{R_{ti}}{L_{ti}}+\mathrm i \omega_0 \right) I_{ti}^{dq} - \frac{V_i^{dq}}{L_{ti}} + \frac{V_{ti}^{dq}}{L_{ti}} 
\end{aligned}
\right.
\end{equation}
\normalsize
where quantities $V^{dq}_i$ and $I^{dq}_{ti}$ represent the $i$-th PCC voltage and
filter current, respectively, $V^{dq}_{ti}$ is the command input to the corresponding
VSC, while $R_{ti}$, $L_{ti}$ and $C_{ti}$ are the converter parameters. Moreover, $V^{dq}_{j}$ is the voltage at the
PCC of each neighboring DGU $j\in\mathcal{N}_i$ and $R_{ij}$ and $L_{ij}$ are,
respectively, the resistance
and impedance of the three-phase power line connecting DGUs $i$ and $j$. 
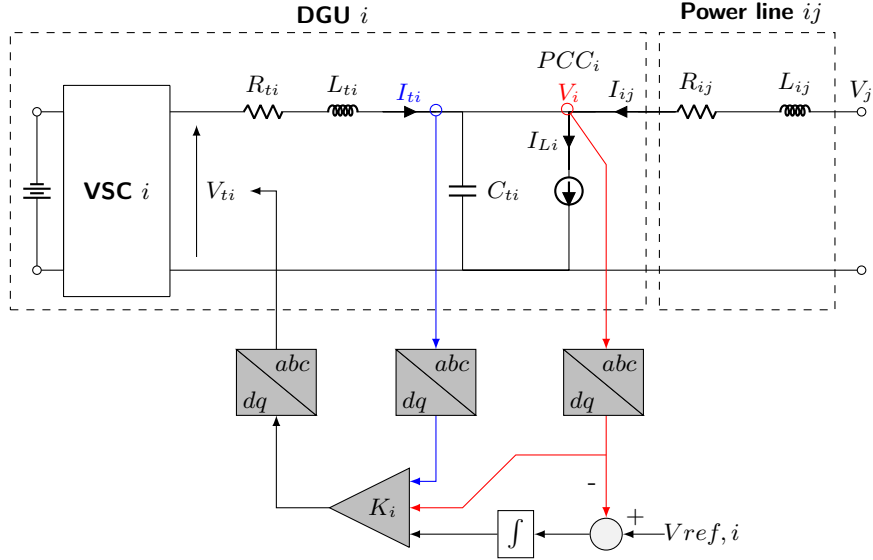
\begin{figure}[!htb]
	\centering
	\ctikzset{bipoles/length=.6cm}
	\tikzstyle{every node}=[font=\sffamily\small]
	\begin{circuitikz}[scale=0.7]
		\draw (1,1)  to [battery, o-o](1,4)
		to [short](1.5,4)
		to [short](1.5,4.5)
		to [short](3.5,4.5)
		to [short](3.5,0.5)
		to [short](1.5,0.5)
		to [short](1.5,4)
		to [short](1.5,1)
		to [short](1,1);
		\node at (2.5,2.5){ \textbf{VSC $i$}};
		\draw[-latex] (4,1.25) -- (4,3.75)node[midway,right]{$V_{ti}$};
		\draw (3.5,4) to [short](4,4)
		to [short](4.5,4)
		to [R=$R_{ti}$] (6,4)
		to [L=$L_{ti}$] (7.5,4)
		to [short, i=$\textcolor{blue}{I _{ti}}$, -] (8.5,4)
		to [short](9,4) 
		to [C, l=$C_{ti}$, -] (9,1)
		to [short](4,1)
		to [short](3.5,1);
		\draw (12.3,4)  to [R=$R_{ij}$] (14.5,4)
		to [L=$L_{ij}$] (16,4)
		to [short, -o] (16.5,4) node[anchor=north,above]{$V_j$};
		\draw (8.5,4) to (11,4) 
		to [ I ] (11 ,1)
		to [short] (9,1)
		to [short, -o] (16.5,1); 
		\draw (11,4) to [short](11.5,4);
		\draw (11,4) node[anchor=north, above]{$\textcolor{red}{V_i}$}  to [short, i_=$I_{Li}$](11,2.9);
		\node at (11,4.6)[anchor=north, above]{$PCC_i$} ;
		\draw (11,4) to [short,i<=$I_{ij}$] (13,4) -- (13,4); 
		\draw[black, dashed] (.5,.25) -- (12.45,.25) -- (12.45,5.5) -- (.5,5.5)node[sloped, midway, above]{{ \textbf{DGU $i$}}}  -- (.5,.25);
		\draw[black, dashed] (12.7,.25) -- (16,.25) -- (16,5.5) -- (12.7,5.5)node[sloped, midway, above]{{ \textbf{Power line $ij$}}}  -- (12.7,.25);
		\draw[red,o-latex] (10.9,4.15) -- (11.7,3) to (11.7,-.5);
		\draw[red,latex-](8,-3.5)-- (9,-3.5) --  (10,-2.5)-- (11.7,-2.5);
		\draw[blue,o-latex] (8.5,4.15) to (8.5,-0.5);
		\draw[blue, -latex] (8.5,-1.75) -- (8.5,-3) to (8,-3);
		\draw (10,-4) node(a) [black, draw,fill=white!20] {$\normalsize{\int}$};
		\draw[-latex] (a.west) to (8,-4);
		\draw (11.7,-4) node(b)[ circle, draw=black, minimum size=12pt, fill=lightgray!20]{};
		\draw[red, -latex] (11.7,-1.75)  -- (b.north) node[pos=0.7, left]{\textcolor{black}{\normalsize{-}}};
		\draw[latex-] (a.east) -- (b.west);
		\draw[-latex] (12.8,-4)  -- (b.east) node[pos=0.7,above]{{+}} node[pos=0.25,right]{$V{ref,i}$};
		\draw[fill=lightgray] (4.75,-0.5)    -- (6.25,-1.75) -- (4.75,-1.75) -- (4.75,-0.5);
		\draw[fill=lightgray] (6.25,-1.75) -- (6.25,-0.5)    -- (4.75,-0.5) -- (6.25,-1.75);
		\node at (5.1,-1.5) {$dq$}; 
		\node at (5.8,-.75) {$abc$}; 
		\draw[fill=lightgray] (7.75,-0.5)    -- (9.25,-1.75) -- (7.75,-1.75) -- (7.75,-0.5);
		\draw[fill=lightgray] (9.25,-1.75) -- (9.25,-0.5)    -- (7.75,-0.5) -- (9.25,-1.75);
		\node at (8.1,-1.5) {$dq$};
		\node at (8.8,-.75) {$abc$}; 
		\draw[fill=lightgray] (10.9,-0.5)    -- (12.4,-1.75) -- (10.9,-1.75) -- (10.9,-0.5);
		\draw[fill=lightgray] (12.4,-1.75) -- (12.4,-0.5)    -- (10.9,-0.5) -- (12.4,-1.75);
		\node at (11.25,-1.5) {$dq$}; 
		\node at (11.95,-.75) {$abc$}; 
		\draw[fill=lightgray] (8,-4.25) -- (8,-2.75) -- (6.5,-3.5) -- (8,-4.25);
		\node at (7.5,-3.5) {$K_i$};
		\draw[-latex] (6.5,-3.5) -- (5.5,-3.5) -- (5.5,-1.75);
		\draw[-latex] (5.5,-.5) -- (5.5,2.5) -- (5,2.5);				   
	\end{circuitikz}
	\caption{Electrical scheme of DGU $i$, power line $ij$, and local
		PnP voltage and frequency controller.}
	\label{fig:ctrl_part}
\end{figure}
Moreover, we recall that \eqref{DGUeqx} has been derived by exploiting
Quasi-Stationary Line (QSL) approximations of power lines \cite{Riverso_TSG}.

Next, we derive the state-space model of the ImG with
dynamics \eqref{DGUeqx}. Notably, we can write
\small
\begin{equation}
\label{eq:subsys_ss}
\text{ ${\Sigma}_{[i]}^{DGU}:$}\left\lbrace
\begin{aligned}
\dot{x}_{[i]}(t) &= A_{ii}{x}_{[i]}(t) +
B_{i}{u}_{[i]}(t)+M_{i}{d}_{[i]}(t)+\xi_{[i]}(t)\\
{y}_{[i]}(t)       &= C_{i}{x}_{[i]}(t)\\
{z}_{[i]}(t)       &= H_{i}{y}_{[i]}(t)\\
\end{aligned}
\right.
\end{equation}
\normalsize
where ${x}_{[i]}=[V_{i}^d, V_{i}^q, I_{ti}^d, I_{ti}^q]^T$ is the state,
${u}_{[i]} = [V_{ti}^d, V_{ti}^q]^T$ the control input,
${d}_{[i]} = [I_{Li}^d,I_{Li}^q]^T$ the exogenous input and
${z}_{[i]} = [V_{i}^d,V_{i}^q]^T$ the controlled variable of the
system. Moreover, we assume that the output $y_{[i]}=
{x}_{[i]}$ is measurable, and let
term $\xi_{[i]}=\sum_{j\in\mathcal{N}_i} A_{ij}( x_{[j]}- x_{[i]})$ accounts
for the coupling with each DGU $j\in\mathcal{N}_i$. 
We highlight that the provided model is identical to the one in
\cite{Riverso_TSG}, except that the
coupling terms have been embedded in the contribution ${\xi}_{[i]}$.
As regards the matrices in \eqref{eq:subsys_ss}, they have the following form:

\begin{equation}
\label{eq:matrices_1}
\renewcommand\arraystretch{2}
A_{ii}=\begin{bmatrix}
0 & \omega_0 & \frac{1}{C_{ti}} &0\\
-\omega_0 & 0 & 0 & \frac{1}{C_{ti}}\\
-\frac{1}{L_{ti}} & 0 & -\frac{R_{ti}}{L_{ti}} &
\omega_0\\
0& -\frac{1}{L_{ti}} & -\omega_0 & -\frac{R_{ti}}{L_{ti}} \\
\end{bmatrix}, A_{ij}=
\frac{1}{C_{ti}}  \begin{bmatrix}
\frac{R_{ij}}{Z^2_{ij}}  & \frac{X_{ij}}{Z^2_{ij}}&\mathbf{0} \\
-\frac{X_{ij}}{Z^2_{ij}}& \frac{R_{ij}}{Z^2_{ij}}  & \mathbf{0}\\
\mathbf{0} & \mathbf{0} & \mathbf{0}\\
\end{bmatrix},
\end{equation}
\begin{equation}
\label{eq:matrices_2}B_{i}=\begin{bmatrix}
\mathbf{0} &\mathbf{0}\\
\frac{1}{L_{ti}} & 0\\
0& \frac{1}{L_{ti}}
\end{bmatrix}, 
M_{i}=\begin{bmatrix}
-\frac{1}{C_{ti}} & 0\\
0 &  -\frac{1}{C_{ti}}\\
\mathbf{0} &\mathbf{0}\\
\end{bmatrix},
C_{i}= I_4,
H_{i}=\begin{bmatrix}
1 & 0 & \mathbf{0}\\
0 & 1 & \mathbf{0}\\
\end{bmatrix},
\end{equation}
\normalsize
where $X_{ij}=\omega_0L_{ij}$ and $Z_{ij }=|R_{ij}+\mathrm i X_{ij}|$.\\
At this point, we can write the ImG model as follows
\begin{equation}
\label{eq:stdformA}
\left\lbrace
\begin{aligned}
\mathbf{\dot{x}}(t) &= \mathbf{Ax}(t) + \mathbf{Bu}(t)+ \mathbf{Md}(t)\\
\mathbf{y}(t)       &= \mathbf{Cx}(t)\\
\mathbf{z}(t)       &= \mathbf{Hy}(t)
\end{aligned}
\right.
\end{equation}
where $\mathbf x = ( x_{[1]},\ldots, x_{[N]})\in\mathbb{R}^{4N}$, $\mathbf u = ( u_{[1]},\ldots, u_{[N]})\in\mathbb{R}^{2N}$, $\mathbf d = ( d_{[1]},\ldots, d_{[N]})\in\mathbb{R}^{2N}$, $\mathbf y = ( y_{[1]},\ldots, y_{[N]})\in\mathbb{R}^{4N}$, $\mathbf z = ( z_{[1]},\ldots, z_{[N]})\in\mathbb{R}^{2N}$. Matrices $\mathbf{A}$,
$\mathbf{B}$, $\mathbf{M}$, $\mathbf C$ and $\mathbf H$, which can be easily derived from \eqref{eq:matrices_1}-\eqref{eq:matrices_2}, are shown in Appendix A.3 of \cite{Riverso2014c}. 
  \section{Design of stabilizing controllers}
  \label{sec:design}
  As in \cite{Riverso_TSG}, in
  order to track constant references $\mathbf{z_{ref}}(t)=\mathbf{\bar
  	z_{ref}}$, the ImG model is augmented with integrators
  (see Figure
  \ref{fig:ctrl_part}, where $z_{ref_{[i]}}=V_{ref,i}$). Hence, we first write the dynamics of the integrators as
  \begin{equation*}
  \begin{aligned}
  {\dot{v}}_{[i]}(t) = {e}_{[i]}(t) &= {z_{ref}}_{[i]}(t)-{z}_{[i]}(t) \\
  &= {z_{ref}}_{[i]}(t)-H_{[i]}C_{[i]}{x}_{[i]}(t),
  \end{aligned}
  \label{eq:intdynamics}
  \end{equation*}
  and then derive the augmented model of the DGU
  \begin{equation}
  \label{eq:modelDGUgen-aug}
  {\hat{\Sigma}}_{[i]}^{DGU}\hspace{-1.5mm}:\hspace{-1mm}
  \left\lbrace
  \begin{aligned}
  {\dot{\hat{x}}}_{[i]}(t) &= \hat{A}_{ii}{\hat{x}}_{[i]}(t) + \hat{B}_{i}{u}_{[i]}(t)+\hat{M}_{i}{\hat{d}}_{[i]}(t)+{\hat\xi}_{[i]}(t)\\
  {\hat{y}}_{[i]}(t)       &= \hat{C}_{i}{\hat{x}}_{[i]}(t)\\
  {z}_{[i]}(t)       &= \hat{H}_{i}{\hat{y}}_{[i]}(t)
  \end{aligned}
  \right.
  \end{equation}
  where ${\hat{x}}_{[i]}=[{x^T}_{[i]},v_{[i]}^T]^T\in\mathbb{R}^6$ is the state,
  ${\hat{y}}_{[i]}={\hat{x}}_{[i]}$
  the measurable output,
  ${\hat{d}}_{[i]}=[{d}_{[i]}^T,{z_{ref}}_{[i]}^T]^T\in\mathbb{R}^4$
  the exogenous signals and
  ${\hat\xi}_{[i]}=\sum_{j\in\mathcal{N}_i}\hat{A}_{ij}({\hat{x}}_{[j]}-{\hat{x}}_{[i]})$. Moreover,
  matrices
  in \eqref{eq:modelDGUgen-aug} have the form

  \begin{equation*}
  \begin{aligned}
  \hat{A}_{ii} &= 
  \left[ \begin{array}{c|cc}
  A_{ii} & \mathbf{0} \\
  \hline
  -H_{i}C_{i}& \mathbf{0}
  \end{array}\right]=
  \left[ \begin{array}{cc|c}
  \hat{\mathcal{A}}_{11,i} & \frac{1}{C_{ti}}I_2 & \mathbf{0}_2 \\
  -\frac{1}{L_{ti}}I_2 & \hat{\mathcal{A}}_{22,i} & \mathbf{0}_2\\
  \hline
  -I_2 &\mathbf{0}_2 & \mathbf{0}_2 \\
  \end{array}\right],
  \\
  \hat{\mathcal{A}}_{11,i}&=\omega_0\begin{bmatrix}
  0 & 1\\
  -1 & 0
  \end{bmatrix}, 
  \hspace{2mm}
  \hat{\mathcal{A}}_{22,i}=\omega_0\begin{bmatrix}
  -\frac{R_{ti}}{L_{ti}} & \omega_0\\
  -\omega_0 & -\frac{R_{ti}}{L_{ti}}
  \end{bmatrix},
  \\
  \hat{A}_{ij}&=\begin{bmatrix}
  A_{ij} &\mathbf{0}\\
  \mathbf{0}&\mathbf{0}
  \end{bmatrix},\hspace{2mm}
  \hat{B}_{i}=\begin{bmatrix}
  B_{i}\\
  \mathbf{0}
  \end{bmatrix},\\
  \hat{C}_{i}&=\begin{bmatrix}
  C_{i} & \mathbf{0}\\
  \mathbf{0} & I_2
  \end{bmatrix},\hspace{2mm}
  \hat{M}_{i}=\begin{bmatrix}
  M_{i} & \mathbf{0} \\
  \mathbf{0} & I_2
  \end{bmatrix},\hspace{2mm}
  \hat{H}_{i}=\begin{bmatrix}
  H_{i} & \mathbf{0}
  \end{bmatrix}.
  \end{aligned}
  \end{equation*}   
  \normalsize 
  
  We highlight that, since all the electrical parameters are positive, the pair
  $(\hat{A}_{ii},\hat{B}_i)$ is controllable (see Proposition 2 in
  \cite{Riverso_TSG}). Therefore, system \eqref{eq:modelDGUgen-aug} can be
  stabilized.
  
  As in \cite{Riverso_TSG}, the overall augmented system is obtained from \eqref{eq:modelDGUgen-aug} as
  \begin{equation}
  \label{eq:sysaugoverall}
  \left\lbrace
  \begin{aligned}
  \mathbf{\dot{\hat{x}}}(t) &= \mathbf{\hat{A}\hat{x}}(t) + \mathbf{\hat{B}u}(t)+ \mathbf{\hat{M}\hat{d}}(t)\\
  \mathbf{\hat{y}}(t)       &= \mathbf{\hat{C}\hat{x}}(t)\\
  \mathbf{z}(t)             &= \mathbf{\hat{H}\hat{y}}(t)
  \end{aligned}
  \right.
  \end{equation}
  where $\mathbf{\hat{x}}$, $\mathbf{\hat{y}}$ and
  $\mathbf{\hat{d}}$ collect variables ${\hat{x}}_{[i]}$,
  ${\hat{y}}_{[i]}$ and ${\hat{d}}_{[i]}$ respectively,
  and matrices $\mathbf{\hat{A}}, \mathbf{\hat{B}},
  \mathbf{\hat{C}}, \mathbf{\hat{M}}$ and
  $\mathbf{\hat{H}}$ are derived directly from the
  systems \eqref{eq:modelDGUgen-aug}. Finally, we equip each DGU ${\hat{\Sigma}}_{[i]}^{DGU}$ with the following state-feedback controller
  \begin{equation*}
  \label{eq:ctrldec}
  {\mathcal{C}}_{[i]}:\qquad {u}_{[i]}(t)=K_{i}{\hat{y}}_{[i]}(t)=K_{i}{\hat{x}}_{[i]}(t),
  \end{equation*}
  where $K_{i}=\mathbb{R}^{2\times 6}$. Note that the computation of
  ${u}_{[i]}$ requires the state of
  ${\hat{\Sigma}}_{[i]}^{DGU}$ only. Hence, we have that the overall control architecture is
  decentralized.
  \subsection{Local conditions implying ImG stability}
  \label{subs:local_conditions}
  If coupling terms $\hat{\xi}_{i}(t)$ are not
  present, the asymptotic stability of the overall ImG can be ensured by simply stabilizing each closed-loop subsystem 
  \begin{equation}
  \label{eq:modelDGUgen-aug-closed}
  \dot{\hat{x}}_{[i]}(t) =\underbrace{(\hat{A}_{ii}+\hat{B}_iK_i)}_{F_i}\hat{x}_{[i]}(t)+\hat{M}_{i}\hat{d}_{[i]},
  \end{equation}
  where, by construction, matrix $F_{i}$ has the following structure
  {\begin{equation}
  	\renewcommand\arraystretch{1.2}
  	\label{eq:Fi}
  	F_{i}=
  	\left[ \begin{array}{ccc}
  	\mathcal{F}_{11,i} & \mathcal{F}_{12,i} & \mathbf{0}_2 \\
  	\mathcal{F}_{21,i} & \mathcal{F}_{22,i} & \mathcal{F}_{23,i}\\
  	-I_2 &\mathbf{0}_2 & \mathbf{0}_2 \\
  	\end{array}\right]. 
  	\end{equation}}\normalsize
  One also has
  \begin{equation}
  \label{eqn:Fi_subplots}
  \mathcal{F}_{11,i}=\hat{\AAA}_{11,i}\hspace{2mm}\text{and}\hspace{2mm}
  \mathcal{F}_{12,i}=\frac{1}{C_{ti}} I_2.
  \end{equation}
  According to Lyapunov theory, system \eqref{eq:modelDGUgen-aug-closed}
  is asymptotically stable if there exists a Lyapunov function
  $\mathcal{V}_i(\hat{x}_{[i]})=\hat{x} _{[i]} ^T P_i\hat{x}_{[i]}$,
  with $P_{i}=P^T_{i}>0$, such that
  
  \begin{equation}
  \label{eq:Lyapeqnith}
  Q_i = (\hat{A}_{ii}+\hat{B}_iK_i)^TP_i + P_i(\hat{A}_{ii}+\hat{B}_iK_i)  = F_{i}^T
  P_{i}+P_{i}F_{i}
  \end{equation}
  \normalsize
  is negative definite.
  
  In a real ImG, however, electric interactions between subsystems
  exist. For this reason, in the following we present the conditions which allows us to guarantee collective ImG stability by designing totally decentralized controllers, even in presence of couplings between DGUs.
  \begin{assum} 
  	\label{ass:ctrl}
  	Each matrix gain $K_{i}$, $i\in\mathcal{D}$
  	is designed using $P_i$ in \eqref{eq:Lyapeqnith} with the following structure
  	\small
  	\begin{equation}
  	\label{eq:pstruct}
  	P_{i}=\left[ \begin{array}{c|c}
  	\eta_i I_{2}& \mathbf{0}_{2} \\
  	\hline
  	\mathbf{0}_{2}  & P_{22,i}\\
  	\end{array}\right],
  	\end{equation}
  	\normalsize
  	where the entries of $P_{22,i}\in\mathbb{R}^{4\times 4}$ are arbitrary and $\eta_i>0$ is a local parameter.
  \end{assum}
  \begin{rmk}
  	\label{rmk:separable}
  	Assumption \ref{ass:ctrl} amounts to use local separable Lyapunov
  	functions in the form $$\mathcal{V}_i( \hat{x}_{[i]})=
  	\eta_i\hat{x}_{[i],1}^T\hat{x}_{[i],1} + \tilde{x}_{[i]}^T\text{
  	}P_{22,i}\tilde{x}_{[i]},$$
  	where $\hat{x}_{[i]}= [\hat{x}_{[i],1}^T,\tilde{x}_{[i]}^T]^T$,
  	$\tilde{x}_{[i]}\in\mathbb{R}^{4}$.
  \end{rmk}
  The second condition regards the parameters $\eta_i$.
  \begin{assum}
  	\label{ass:equal_ratio}
  	Let $\bar\sigma>0$ be a
  	constant parameter, common to
  	all the DGUs. Parameters $\eta_i$ in \eqref{eq:pstruct} are given by
  	\begin{equation*}
  	\label{eq:equal_ratio}
  	\eta_i = \bar\sigma C_{ti}
  	\hspace{7mm} \forall i\in\mathcal{D}.
  	\end{equation*}
  \end{assum}
  The stability analysis continues by showing that, if Assumption~\ref{ass:ctrl} holds, Lyapunov
  theory can certify marginal
  stability (but not asymptotic stability) of \eqref{eq:sysaugoverall}. To this purpose, we provide the
  following Proposition.
  \begin{prp}
  	\label{pr:prop_1}
  	Under Assumption \ref{ass:ctrl}, matrix $Q_i$ in \eqref{eq:Lyapeqnith} cannot be negative
  	definite. Moreover, 
  	\begin{equation}
  	\label{eq:Qi_semidef}
  	Q_i \leq 0
  	\end{equation} 
  	implies that $Q_i$ has the following structure:
  	\small
  	\begin{equation}
  	\renewcommand\arraystretch{1.2}
  	Q_i=
  	\label{eq:qstruct}
  	\left[ \begin{array}{ccc}
  	\mathbf{0}_2 & \mathbf{0}_2 & \mathbf{0}_2 \\
  	\mathbf{0}_2 & \mathcal{Q}_{22,i} & \mathcal{Q}_{23,i}\\
  	\mathbf{0}_2 & \mathcal{Q}_{23,i} & \mathcal{Q}_{33,i}\\
  	\end{array}\right] .
  	\end{equation}
  	\normalsize
  \end{prp}
  \begin{proof}
  	By substituting \eqref{eq:Fi} and \eqref{eq:pstruct} in
  	\eqref{eq:Lyapeqnith}, one gets that the first two diagonal elements of $Q_i$ are zero. This shows that $Q_i$ cannot be negative definite. Moreover, from basic linear
  	algebra, if a negative semidefinite matrix has
  	a zero element on its diagonal, the corresponding row and column
  	have zero entries. Then \eqref{eq:Qi_semidef} implies \eqref{eq:qstruct}. 
  \end{proof} 
  Next, we consider the overall closed-loop ImG model, given by
  \begin{equation}
  \label{eq:sysaugoverallclosed}
  \left\lbrace
  \begin{aligned}
  \mathbf{\dot{\hat{x}}}(t) &= (\mathbf{\hat{A}+\hat{B}K})\mathbf{\hat{x}}(t)+ \mathbf{\hat{M}\hat{d}}(t)\\
  \mathbf{\hat{y}}(t)       &= \mathbf{\hat{C}\hat{x}}(t)\\
  \mathbf{z}(t)       &= \mathbf{\hat{H}\hat{y}}(t)
  \end{aligned}
  \right.
  \end{equation}
  where $\mathbf{K}=\text{diag}(K_{1},\dots,K_{N})$. Being
  $\mathbf{P}=\text{diag}(P_{1},\dots,P_{N})$, the collective
  Lyapunov function is 
  \begin{equation*}
  \label{eq_coll_lyap}
  \mathcal{V}(\mathbf{\hat{x}})=\sum_{i=1}^N\mathcal{V}_i(\hat{x}_{[i]})=\mathbf{\hat{x}}^T\mathbf{P}\mathbf{\hat{x}}.
  \end{equation*}
  Consequently, one has that
  $\dot{\mathcal{V}}(\mathbf{\hat{x}})=
  \mathbf{\hat{x}}^T\mathbf{Q}\mathbf{\hat{x}}$, where
  \begin{equation}
  \nonumber
  \mathbf{Q} = (\mathbf{\hat{A}}+\mathbf{\hat{B}K})^T
  \mathbf{P}+\mathbf{P}(\mathbf{\hat{A}}+\mathbf{\hat{B}K}).
  \end{equation}
  From Proposition \ref{pr:prop_1}, we know that, if Assumption \ref{ass:ctrl} holds, then (i) matrix  $\mathbf{Q}$
  cannot be negative definite, and (ii), at most, one can have
  \begin{equation}
  \label{eq:Lyapeqnoverall}
  \mathbf{Q} \leq 0.
  \end{equation}
  At this point, we notice that, even if \eqref{eq:Qi_semidef} is verified for all $i\in\mathcal{D}$, the
  inequality \eqref{eq:Lyapeqnoverall} might be violated because of the nonzero
  coupling terms $\hat{A}_{ij}$ in matrix $\mathbf{\hat A}$ (an example is provided in Appendix B of \cite{Riverso2014c}).
  However, through the next Proposition, we show that \eqref{eq:Lyapeqnoverall} is always satisfied if also Assumption \ref{ass:equal_ratio} is fulfilled.
  \begin{prp}
  	\label{pr:semidefinite_abc}
  	Under Assumptions \ref{ass:ctrl} and \ref{ass:equal_ratio}, if
  	matrix gains $K_i$ are computed to
  	satisfy \eqref{eq:Qi_semidef} for all $i\in\mathcal{D}$, then \eqref{eq:Lyapeqnoverall} holds.
  \end{prp}
  \begin{proof}
  	We start by decomposing the matrix $\mathbf{\hat
  		A}$ as follows
  	\begin{equation}
  	\label{eq:decomposition}
  	\mathbf{\hat A} = \mathbf{\hat A_D}+\mathbf{\hat A_{\Xi}} + \mathbf{\hat A_C},
  	\end{equation}
  	where (i) $\mathbf{\hat{A}_{D}}=\text{diag}(\hat{A}_{ii},\dots,\hat{A}_{NN})$
  	collects the local dynamics only, (ii) $\mathbf{\hat{A}_{\Xi}}=\text{diag}(\hat{A}_{\xi 1},\dots,\hat{A}_{\xi
  		N})$ with 
  	\small
  	\begin{equation}
  	\label{eq:A_xi}
  	\renewcommand\arraystretch{1.5}
  	\hat{A}_{\xi i}=	\frac{1}{C_{ti}}\left[\begin{array}{cc|c}
  	\sum\limits_{j \in \mathcal{N}_i}
  	-\widetilde{R}_{ij}&\sum\limits_{j \in
  		\mathcal{N}_i}-\widetilde{X}_{ij}&\mathbf{0}\\
  	\sum\limits_{j \in
  		\mathcal{N}_i}\widetilde{X}_{ij}&\sum\limits_{j \in
  		\mathcal{N}_i}-\widetilde{R}_{ij}&\mathbf{0}\\
  	\hline
  	\mathbf{0} & \mathbf{0} &\mathbf{0}\\
  	\end{array}\right],
  	\end{equation}
  	\normalsize
  	$\widetilde{R}_{ij} = \frac{R_{ij}}{Z^2_{ij}}$ and $\widetilde{X}_{ij} = \frac{X_{ij}}{Z^2_{ij}}$, takes into account the
  	dependence of each local state on the neighboring DGUs, and (iii)
  	$\mathbf{\hat{A}_C}$ includes the effect of couplings. Notably, this latter matrix is composed by zero
  	blocks on the diagonal and blocks $\hat{A}_{ij}$, $i\neq j$ outside
  	the diagonal.\\
  	Our goal is to demonstrate \eqref{eq:Lyapeqnoverall}, which, using
  	\eqref{eq:decomposition}, becomes
  	\small
  	\begin{equation}
  	\begin{aligned}
  	\label{eq:Lyap_abc}
  	\underbrace{\mathbf{(\hat{A}_{D}+\hat{B}K)^TP+
  			P(\hat{A}_{D}+\hat{B}K)}}_{({a})}&+\underbrace{\mathbf{\hat{A}_{\Xi}^TP}+\mathbf{P\hat{A}_{\Xi}}}_{(b)}+\underbrace{\mathbf{\hat{A}_{C}}^T
  		\mathbf{P+P\hat{A}_{C}}}_{(c)}\leq 0.
  	\end{aligned}
  	\end{equation}
  	\normalsize
  	Since \eqref{eq:Qi_semidef} holds, we have that $(a) =
  	\mathrm{diag}(Q_1,\dots,Q_N)\leq 0$. At this point, we need to study
  	the contribution of matrix $(b)+(c)$ in \eqref{eq:Lyap_abc}. By
  	construction (recalling \eqref{eq:pstruct} and \eqref{eq:A_xi}),
  	matrix $(b)$ is block
  	diagonal, collecting, on its diagonal, blocks ${\hat{A}}_{\xi i}^TP_{i}+P_{i}{\hat{A}}_{\xi i}$ in the form
  	\begin{equation}\renewcommand\arraystretch{2}
  	\label{eq:element_of_b}
  	\begin{aligned}
  	&\left[\begin{array}{cc|c}
  	-\frac{\eta_{i}}{C_{ti}}\sum\limits_{j \in \mathcal{N}_i}\left(\widetilde{R}_{ij}+\widetilde{R}_{ij}\right)& \frac{\eta_{i}}{C_{ti}}\sum\limits_{j \in \mathcal{N}_i}\left(\widetilde{X}_{ij} -\widetilde{X}_{ij}\right)& \mathbf{0} \\
  	\frac{\eta_{i}}{C_{ti}}\sum\limits_{j \in \mathcal{N}_i}\left(\widetilde{X}_{ij} -\widetilde{X}_{ij}\right)&  -\frac{\eta_{i}}{C_{ti}}\sum\limits_{j \in \mathcal{N}_i}\left(\widetilde{R}_{ij}+\widetilde{R}_{ij}\right)& \mathbf{0} \\
  	\hline
  	\mathbf{0} & \mathbf{0} & \mathbf{0}\\
  	\end{array}\right]=\\
  	&\left[\begin{array}{cc|c}
  	-\sum\limits_{j \in \mathcal{N}_i}
  	2\tilde{\eta}_{ij}&0&\mathbf{0}\\
  	0&-\sum\limits_{j \in
  		\mathcal{N}_i}2\tilde{\eta}_{ij}&\mathbf{0}\\
  	\hline
  	\mathbf{0} & \mathbf{0} &\mathbf{0}\\
  	\end{array}\right],
  	\end{aligned}
  	\end{equation}
  	with $\tilde{\eta}_{ij}=\frac{\eta_{i}}{C_{ti}}\widetilde{R}_{ij}=\bar{\sigma}\widetilde{R}_{ij}$.
  	Regarding matrix $(c)$, we have that each the block in
  	position $(i,j)$ is equal to  
  	\begin{equation}
  	\label{eqn:Aij_offdiag}
  	\left\{ \begin{array}{ll}
  	P_i\hat{A}_{ij}+\hat{A}_{ji}^TP_j & \hspace{7mm}\mbox{if } j\in\mathcal{N}_i \\
  	\mathbf{0} & \hspace{7mm}\mbox{otherwise}
  	\end{array} \right.
  	\end{equation}
  	In particular, recalling Assumption \ref{ass:equal_ratio}, by direct calculation, it results

  	\begin{equation}\renewcommand\arraystretch{2}
  	\label{eq:element_of_c}
  	\begin{aligned}
  	P_i\hat{A}_{ij}+\hat{A}_{ji}^TP_j &=
  	\left[\begin{array}{cc|c}
  	\frac{\eta_{i}}{C_{ti}}\widetilde{R}_{ij}+\frac{\eta_{j}}{C_{tj}}\widetilde{R}_{ji}& 0 &\mathbf{0}\\
  	0 &  \frac{\eta_{i}}{C_{ti}}\widetilde{R}_{ij}+\frac{\eta_{j}}{C_{tj}}\widetilde{R}_{ji}&\mathbf{0}\\
  	\hline
  	\mathbf{0} & \mathbf{0} & \mathbf{0}\\
  	\end{array}\right]= \\
  	&=\left[\begin{array}{cc|c}
  	\tilde\eta_{ij}+\tilde\eta_{ji} & 0 & \mathbf{0} \\
  	0 &  \tilde\eta_{ij}+\tilde\eta_{ji} & \mathbf{0} \\
  	\hline
  	\mathbf{0} & \mathbf{0} & \mathbf{0}\\
  	\end{array}\right]= \left[\begin{array}{c|c}
  	2\tilde\eta_{ij}I_2 & \mathbf{0} \\
  	\hline
  	\mathbf{0} & \mathbf{0}\\
  	\end{array}\right].
  	\end{aligned}
  	\end{equation}
  	\normalsize
  	By looking at \eqref{eq:element_of_b} and
  	\eqref{eq:element_of_c}, we observe that only the elements in position $(1,1)$ and $(2,2)$ of each
  	$6\times 6$ block of $(b)+(c)$ can be different from zero. Therefore, the positive/negative definiteness of the $6N\times 6N$
  	matrix $(b)+(c)$ can be equivalently studied by considering the
  	$2N\times 2N$ matrix
  	\begin{equation} 
  	\label{eq:laplacian}
  	\mathcal{L} = \left[\begin{array}{cccc}
  	{\Phi}_{11}&\Phi_{12} & \dots  & \Phi_{1N} \\
  	\Phi_{21} & \ddots &\ddots & \vdots \\
  	\vdots &\ddots & {\Phi}_{N-1\text{ }N-1} &
  	\Phi_{N-1\text{ }N}  \\
  	\Phi_{N1}  & \dots  & \Phi_{N\text{ }N-1} & {\Phi}_{NN}
  	\end{array}
  	\right],
  	\end{equation}
  	obtained by deleting the last four rows and columns in each block
  	of $(b)+(c)$. In particular, we can write \eqref{eq:laplacian} as $\mathcal{L} =
  	\mathcal{M}+\mathcal{G}$, where 
  	\small
  	\begin{equation*}
  	\mathcal{M}=\text{diag}({\Phi}_{11} ,\dots,{\Phi}_{NN}),\hspace{2mm}
  	\Phi_{ii} = \begin{bmatrix}
  	\sum\limits_{j\in\mathcal{N}_i}-2\tilde\eta_{ij}
  	& 0 \\
  	0 & \sum\limits_{j\in\mathcal{N}_i}-2\tilde\eta_{ij} \\
  	\end{bmatrix},
  	\end{equation*}
  	\normalsize and
  	\small
  	\begin{equation*}
  	\label{eq:GG}
  	\mathcal{G}=\left[\begin{array}{cccc}
  	\mathbf{0}&\Phi_{12} & \dots  & \Phi_{1N}\\
  	\Phi_{21} & \mathbf{0} & \ddots  & \vdots \\
  	\vdots &\ddots & \ddots  & \Phi_{N-1\text{ }N}\\
  	\Phi_{N1}  & \dots & \Phi_{N\text{ }N-1}  & \mathbf{0}
  	\end{array}
  	\right].
  	\end{equation*}
  	\normalsize
  	We highlight that, from \eqref{eqn:Aij_offdiag} and
  	\eqref{eq:element_of_c}, blocks $\Phi_{ij}$, $i\neq j$, are equal to 
  	\begin{equation*}
  	\label{eq:bar_eta}
  	\Phi_{ij}=\left\{ \begin{array}{ll}
  	\begin{bmatrix}
  	2\tilde\eta_{ij}
  	& 0 \\
  	0 & 2\tilde\eta_{ij} \\
  	\end{bmatrix}
  	& \hspace{7mm}\mbox{if }
  	j\in\mathcal{N}_i \\
  	\mathbf{0}_2 & \hspace{7mm}\mbox{otherwise}
  	\end{array} \right.
  	\end{equation*}
  	Next, we notice that $\mathcal{L}$ is
  	symmetric, with non negative off-diagonal elements and zero row and
  	column sum. In other words, $\mathcal{L}$ is a Laplacian matrix
  	\cite{godsil2001algebraic}, and, as such, it is
  	negative semidefinite. This allows us to show that \eqref{eq:Lyap_abc}
  	(and, equivalently, \eqref{eq:Lyapeqnoverall}) holds. 
  \end{proof}
  \begin{rmk}
  	The proof of Proposition \ref{pr:semidefinite_abc} reveals that, under
  	Assumptions \ref{ass:ctrl} and \ref{ass:equal_ratio}, interactions between local Lyapunov functions
  	$\mathcal{V}_i(\hat{x}_{[i]})$ due to terms $\hat{A}_{ij}$, $i\neq j$, take
  	the form of a weighted Laplacian matrix associated with the graph $\mathcal{G}_{el}$. Furthermore,
  	differently from the idea in \cite{Riverso_TSG} of nullifying interactions by
  	choosing $\eta_i>0$ in \eqref{eq:pstruct} sufficiently small, here \eqref{eq:Lyapeqnoverall} holds
  	true even if parameters $\eta_i$ are large.
  \end{rmk}

  In order to show the asymptotic stability of the ImG, we will complement 
  Proposition \ref{pr:semidefinite_abc} with the application of the LaSalle invariance
  theorem. This will be done in Section \ref{subsec:localdesign}, after providing the algorithm for local control design.

  \subsection{Design of local controllers}
  \label{subsec:localdesign}
  
  For applying Theorem~\ref{thm:overall_stability}, local control gains $K_i$, $i\in\DD$ guaranteeing \eqref{eq:Qi_semidef} must be designed.
  Let us parametrize the unknown quantities in \eqref{eq:Lyapeqnith} as follows
  \begin{equation}
  \label{eq:parametrization}
  P_i = Y_i^{-1},\quad K_i = G_iY_i^{-1},
  \end{equation}
  where  $
  \renewcommand\arraystretch{1.2}
  G_{i}= \left[ \begin{array}{ccc} \mathcal{G}_{11,i} & \mathcal{G}_{12,i} & \mathcal{G}_{13,i}\end{array}\right]\in \mathbb{R}^{2\times 6}$, $\mathcal{G}_{11,i}, \mathcal{G}_{12,i}, \mathcal{G}_{13,i} \in \mathbb{R}^{2\times 2}$ and, under Assumption \ref{ass:ctrl},
  \small{\begin{equation}
  	\label{eq:Yi_param}
  	\renewcommand\arraystretch{1.2}
  	Y_{i}=
  	\left[ \begin{array}{c|cc}
  	\eta_i^{-1}I_2 & \mathbf{0}_2 & \mathbf{0}_2 \\
  	\hline
  	\mathbf{0}_2 & \mathcal{Y}_{22,i} & \mathcal{Y}_{23,i}\\
  	\mathbf{0}_2 & \mathcal{Y}_{23,i}^T & \mathcal{Y}_{33,i} \\
  	\end{array}\right] = \left[ \begin{array}{c|cc}
  	\eta_i^{-1}I_2 & \mathbf{0} \\
  	\hline
  	\mathbf{0}& Y_{22,i}
  	\end{array}\right].
  	\end{equation}}\normalsize 
  We also focus on the matrix
  \begin{equation}
  \label{eq:Qtilde}
  \tilde Q_i=Y_i Q_i Y_i
  \end{equation}
  instead of $Q_i$. Apparently, when $Y_i>0$, $Q_i$ is negative semidefinite if and only if $\tilde Q_i$ has the same property.
  The advantage of the parametrization \eqref{eq:parametrization} is that, as shown below, all entries of $\tilde Q_i$ depend linearly on $G_i$ and $Y_i$, while products of matrices $P_i$ and $K_i$ appear in $Q_i$ (see \eqref{eq:Lyapeqnith}). Indeed, \eqref{eq:parametrization} is routinely used for mapping the design of state-feedback gains into LMIs \cite{Boyd1994}.
  
  By direct calculation, from
  \eqref{eq:parametrization}, \eqref{eq:Qtilde}, and \eqref{eq:Lyapeqnith} one has 
  \begin{equation}
  \label{tildeQ_i}
  \tilde{Q}_i = \left[ 
  \begin{array}{ccc}
  \mathbf{0}_2 & \tilde{\mathcal{Q}}_{12,i} & \tilde{\mathcal{Q}}_{13,i} \\
  \tilde{\mathcal{Q}}_{12,i}^T & \tilde{\mathcal{Q}}_{22,i} & \tilde{\mathcal{Q}}_{23,i}\\
  \tilde{\mathcal{Q}}_{23,i}^T & \tilde{\mathcal{Q}}_{23,i}^T & \mathbf{0}_2 \\
  \end{array}
  \right],
  \end{equation}
  where
  
  \begin{align}
  \label{eq:cQ12_tilde}
  \tQQ_{12,i}&=\frac{1}{C_{ti}}\YY_{22,i}-\frac{1}{L_{ti}\eta_i} I_2+ \frac{1}{L_{ti}}\GG_{11,i}^T \\
  \label{eq:cQ22_tilde}
  \tQQ_{22,i}&=\hat{\mathcal{A}}_{22,i}\mathcal{Y}_{22,i} + \mathcal{Y}_{22,i} \hat{\mathcal{A}}_{22,i}^T+\frac{1}{L_{ti}}\mathcal{G}_{12,i} + \frac{1}{L_{ti}}\mathcal{G}_{12,i}^T \\
  \label{eq:cQ13_tilde}
  \tilde{\mathcal{Q}}_{13,i} &= \frac{1}{C_{ti}}\mathcal{Y}_{23,i} - \eta_i^{-1}I \\ 
  \label{eq:cQ23_tilde}	
  \tilde{\mathcal{Q}}_{23,i}&=\hat{\mathcal{A}}_{22,i}\mathcal{Y}_{23,i} + \frac{1}{L_{ti}}\mathcal{G}_{13,i}.
  \end{align}
  \normalsize
  
  The following result mirrors Proposition \ref{pr:prop_1} and provides key properties of the blocks of $\tilde Q_i$ that will be used for control design.
  \begin{prp}
  	\label{pr:Qtilde_semidef}
  	Under Assumption \ref{ass:ctrl}, 
  	\begin{enumerate}[(i)]
  		\item \label{Qtilde_a} the matrix $\tilde{Q}_i$ in \eqref{tildeQ_i} cannot be negative definite;
  		\item \label{Qtilde_b} $\tilde{Q}_i\leq 0$ implies
  		\begin{equation}
  		\label{eq:Qi_tilde_struct}
  		\tilde{Q}_i  = \left[ \begin{array}{ccc}
  		\mathbf{0}_2 & \mathbf{0}_2 & \mathbf{0}_2 \\
  		\mathbf{0}_2 & \tilde{\mathcal{Q}}_{22,i} & \mathbf{0}_2\\
  		\mathbf{0}_2 & \mathbf{0}_2 & \mathbf{0}_2 \\
  		\end{array}\right] ,
  		\end{equation}
  		where $\tQQ_{22,i}$ is given in \eqref{eq:cQ22_tilde}.
  	\end{enumerate}
  \end{prp}
  \begin{proof}
  	Point \eqref{Qtilde_a} follows from the fact that, by construction, the upper left block of $\tilde{Q}_i$ is zero (see \eqref{tildeQ_i}). As for point \eqref{Qtilde_b}, we exploit the same argument used in the proof of Proposition \ref{pr:prop_1}. 
  \end{proof}
  
  Hereafter, we focus on the design of the local control gains. Using the notation of Section \ref{subs:local_conditions}, the synthesis problem can be summarized as follows.
  \begin{prbl}
  	\label{prbl:problem_1}
  	Under Assumptions \ref{ass:ctrl} and \ref{ass:equal_ratio}, compute matrices $K_i$, $i\in\DD$ and $P_{22,i}>0$ such that  \eqref{eq:Qi_semidef} holds.
  \end{prbl}
  If $P_{22,i}>0$, then $Y_i>0$ and hence $Q_i\leq 0 \Leftrightarrow \tilde Q_i\leq 0$.
  By means of Proposition \ref{pr:Qtilde_semidef}-\eqref{Qtilde_b}, the last inequality is equivalent to 
  \begin{align}
  \tQQ_{22,i}&\leq 0 \label{ineq_tQQ22} \\
  \tQQ_{13,i}&= 0, ~\tQQ_{23,i}= 0, \mbox{ and } \tQQ_{12,i}= 0 \label{ineq_tQQ132312}
  \end{align}
  Therefore, Problem \ref{prbl:problem_1} can be rephrased as follows.
  \begin{prbl}
  	\label{prbl:problem_2}
  	Under Assumptions \ref{ass:ctrl} and \ref{ass:equal_ratio}, compute matrices $Y_{22,i}>0$ 
  	and $G_i$, $i\in\DD$ such that \eqref{ineq_tQQ22} and \eqref{ineq_tQQ132312} hold.
  \end{prbl}
  We highlight that the inequality \eqref{ineq_tQQ22} can be always verified strictly, as shown in the following proposition.
  
  \begin{prp}
  	\label{pr:tQ22_neg_def}
  	For any positive-definite matrix $\Gamma_{i}\in\Rset^{2\times 2}$ there exist matrices $\GG_{12,i}$ and $\YY_{22,i}>0$ verifying 
  	\[
  	\begin{aligned}
  	\tQQ_{22,i}&\leq -\mathcal{Y}_{22,i} \Gamma_{i}^{-1}\mathcal{Y}_{22,i}	\label{eq:new_Lyapunov}
  	\end{aligned}
  	\]
  \end{prp}
  \begin{proof}
  	We first show that, $\forall~\Gamma_i>0$, \eqref{eq:new_Lyapunov} is equivalent to
  	\small
  	\begin{align}
  	\label{eq:Lyapunov}
  	& (\hat{\mathcal{A}}_{22,i} + \tilde{K}_{22,i})^T\mathcal{\tilde{P}}_{22,i} + \mathcal{\tilde{P}}_{22,i}(\hat{\mathcal{A}}_{22,i}+\tilde{K}_{22,i}) \leq -\Gamma_{i}^{-1} \\
  	\label{eq tK22}
  	& \tilde{K}_{22,i}=\frac{1}{L_{ti}}\mathcal{G}_{12,i}\mathcal{\tilde{P}}_{22,i}
  	\end{align}
  	\normalsize
  	where $\mathcal{\tilde{P}}_{22,i} = \mathcal{Y}_{22,i}^{-1}$. 
  	Notably, \eqref{eq:Lyapunov} is obtained pre- and post- multiplying the expression of $\tQQ_{22,i}$ in \eqref{eq:cQ22_tilde} by  $\mathcal{Y}_{22,i}^{-1}$. Since the pair $(\hat{\mathcal{A}}_{22,i}, I_2)$ is controllable, one can always find a matrix $\tilde{K}_{22,i}$ that stabilizes $(\hat{\mathcal{A}}_{22,i} + \tilde{K}_{22,i})$. Therefore, for any $\Gamma_i>0$, there exists  $\mathcal{\tilde{P}}_{22,i} = \mathcal{\tilde{P}}_{22,i}^T > 0$ verifying the Lyapunov equation given by \eqref{eq:Lyapunov} with the equality sign. Hence, the inequality $\eqref{eq:new_Lyapunov}$ is verified by using in \eqref{eq:cQ22_tilde} the matrix $\mathcal{G}_{12,i}$ computed from $\tilde{K}_{22,i}$ via \eqref{eq tK22}. 
  \end{proof}
  We notice that, from  \eqref{eq:new_Lyapunov}, the matrix $\Gamma_i^{-1}$ can be interpreted as a robustness margin in the fulfillment of the inequality $\YY_{22,i}^{-1}\tQQ_{22,i}\YY_{22,i}^{-1}\leq 0$.
  
  Next, we discuss equations \eqref{ineq_tQQ132312}. From \eqref{eq:cQ22_tilde}, \eqref{eq:cQ13_tilde}, and \eqref{eq:cQ12_tilde} they are equivalent, respectively, to
  \begin{align}
  \YY_{23,i}&=\frac{C_{ti}}{\eta_i}I_2, \label{eq:Y23} \\
  \GG_{13,i}&=-\frac{L_{ti}C_{ti}}{\eta_i}\hat{\mathcal{A}}_{22,i}, \label{eq:G13} \\
  \GG_{11,i}&=\frac{L_{ti}}{C_{ti}}\mathcal{Y}_{22,i} + \frac{1}{\eta_i}I_2. \label{eq:G11}
  \end{align}
  Therefore, Problem~\ref{prbl:problem_2} can be stated in the following final form.
  \begin{prbl}
  	\label{prbl:problem_3}
  	Under Assumptions \ref{ass:ctrl} and \ref{ass:equal_ratio}, for a given matrix $\Gamma_i>0$, find $\YY_{22,i}=\YY_{22,i}^T>0$, $\YY_{33,i}=\YY_{33,i}^T>0$ and $\GG_{12,i}$ verifying \eqref{eq:new_Lyapunov} and $Y_{22,i}>0$ (with block $\YY_{23,i}$ in \eqref{eq:Yi_param} given by \eqref{eq:Y23}).
  \end{prbl}
  After solving Problem \ref{prbl:problem_3}, blocks $\GG_{13,i}$ and  $\GG_{11,i}$ of the matrix $G_i$ can be computed as in \eqref{eq:G13} and \eqref{eq:G11}. Furthermore, the controller $K_i$ can be recovered from \eqref{eq:parametrization}.
  
  We highlight that, from Proposition \ref{pr:tQ22_neg_def}, Problem \ref{prbl:problem_3} can be always solved and, therefore, controllers $K_i$ always exist. Most importantly, as shown in the next theorem,
  this control synthesis algorithm guarantees asymptotic stability of the ImG.
  
  \begin{theorem}
  	\label{thm:overall_stability}
  	If Assumptions \ref{ass:ctrl} and \ref{ass:equal_ratio} are fulfilled, local control gains $K_i$ $ i\in\mathcal{D}$ are computed by solving Problem  \ref{prbl:problem_3}, and
  	the graph $\mathcal{G}_{el}$ is connected, then system \eqref{eq:sysaugoverallclosed} is asymptotically stable.
  \end{theorem}
  
  The proof of Theorem \ref{thm:overall_stability} is provided in the Appendix.

  \subsection{Computation of local controllers through numerical
  	optimization}
  \label{sec:ctrls_computation}
  
  It is not difficult to see that a solution to Problem~\ref{prbl:problem_3} is provided by the following LMI optimization problem 
  \begin{subequations}
  	\label{eq:optproblem}
  	\begin{align}
  	\mathcal{O}_i: &\min_{\substack{\YY_{22,i},\YY_{33,i},\GG_{12,i}, \\ \gamma_{1i}, \gamma_{2i}, \beta_{i},\zeta_{i}}}\quad \alpha_{1i}\gamma_{1i}+\alpha_{2i}\gamma_{2i} + \alpha_{3i}\beta_{i}+\alpha_{4i}\zeta_{i}\nonumber \\
  	\label{eq:Ystruct}&Y_{i}=\left[ \begin{matrix}
  	\eta_i^{-1}I_2 & \mathbf{0}_2 & \mathbf{0}_2 \\
  	\mathbf{0}_2 & \mathcal{Y}_{22,i} & \frac{C_{ti}}{\eta_i}I_2\\
  	\mathbf{0}_2 & \frac{C_{ti}}{\eta_i}I_2 & \mathcal{Y}_{33,i} \\
  	\end{matrix}\right]>0\\
  	\label{eq:LMIstab}&\left[\begin{matrix}
  	\hat{\mathcal{A}}_{22,i}\mathcal{Y}_{22,i} + \mathcal{Y}_{22,i} \hat{\mathcal{A}}_{22,i}^T+\frac{1}{L_{ti}}\mathcal{G}_{12,i} + \frac{1}{L_{ti}}\mathcal{G}_{12,i}^T & \mathcal{Y}_{22,i} \\
  	\mathcal{Y}_{22,i} & -\Gamma_{i}
  	\end{matrix}\right]\le 0\\
  	\label{eq:Gcostr}&\begin{bmatrix}
  	-\beta_{i}I & G_{i}^T\\
  	G_{i} & -I
  	\end{bmatrix}<0,\hspace{2mm}\begin{bmatrix}
  	Y_{i} & I\\
  	I & \zeta_{i}I
  	\end{bmatrix}>0\\
  	&\gamma_{1i}> 0,\quad \gamma_{2i}> 0, \quad \beta_{i}>0,\quad\zeta_{i}>0
  	\end{align}
  \end{subequations}
  where $\alpha_{ki}$, $k = 1, \dots, 4$ are positive weights and $\Gamma_i=\diag(\gamma_{1i}, \gamma_{2i})$.
  Indeed, \eqref{eq:LMIstab} corresponds, up to Schur complement, to \eqref{eq:new_Lyapunov}. Furthermore, constraints \eqref{eq:Gcostr} are always feasible. Their role is to penalize aggressive control actions \cite{Riverso_TSG} since they correspond to the bound $||K_i||_2\leq\sqrt{\beta_i}\zeta_i$ and large values of $\beta_i$ and $\zeta_i$ are are penalized in the cost function. Finally, the minimization of $\gamma_{1i}$ and $\gamma_{2i}$ corresponds to the maximization of the robustness margin $\Gamma_i^{-1}$ discussed in the previous section.

  We highlight that the computation
  of controller $\mathcal{C}_{[i]}$ is completely decentralized
  (i.e. independent of the synthesis of controllers
  $\mathcal{C}_{[j]}$, $j\neq i$), as constraints in \eqref{eq:optproblem} depend upon local electrical parameters of DGU $i$ and local design
  parameters ($\alpha_{ki}$, $k=1,\dots,4$). 

  \begin{rmk}
  	Controllers $\mathcal{C}_{[i]}$ yield a closed-loop DGU model
  	that is linear. Hence, it can be easily complemented with pre-filters
  	(for tuning the local bandwidth) and load-current compensators. These
  	enhancements (not used in the simulation in Section \ref{sec:simulations}) are described in \cite{Riverso2014c}.
  \end{rmk}
  \subsection{PnP operations}
  \label{sec:PnP}
  In this Section, we describe the operations to be
  performed whenever the plug-in/-out of a DGU is
  required, in order to preserve the overall ImG stability.
  
  \textbf{Plug-in operation.}
  Suppose that we want to connect a new DGU
  (say $\hat{\Sigma}_{[N+1]}^{DGU}$) to the ImG
  and let $\mathcal{N}_{[N+1]}$ be the
  set of DGUs that will be directly connected to
  ${\hat{\Sigma}}_{[N+1]}^{DGU}$ through
  power lines. 
  
  Then, in order to preserve stability, it is enough to equip $\hat{\Sigma}_{[N+1]}^{DGU}$ with the controller obtained by solving the LMI problem $\OO_{N+1}$ and connect the DGU to the ImG. In particular, differently from the procedure in \cite{Riverso_TSG}, the plug-in of a DGU is never denied and DGUs in
  $\mathcal{N}_{[N+1]}$ do not have to retune their local
  controllers.

  %
  
  \textbf{Unplugging operation.}
  When a DGU is disconnected, this has no impact on the controllers of
  the remaining units, if they are designed using the line-independent
  method described in Section \ref{sec:ctrls_computation}. Therefore, in
  view of Theorem \ref{thm:overall_stability}, stability of the ImG is
  preserved. 
  \begin{rmk}
  	\label{rmk:connected_graph}
  	In Theorem \ref{thm:overall_stability}, the connectivity of $\mathcal{G}_{el}$ is assumed to prove that the origin of $\eqref{eq:sysaugoverallclosed}$ is asymptotic stable. On the other hand, we should notice that unplugging operations might make $\mathcal{G}_{el}$ disconnected. These events, however, do not affect the main results since the same analysis presented in Section \ref{sec:design} can be conducted with respect to each connected subgraph $\mathcal{G}^s_{el}$ (which, in turn, can be seen as an independent ImG) arising from the unplugging operations. In Section \ref{sec:simulations}, we show, through simulations, the capability of PnP line-independent regulators to preserve voltage and frequency stability in spite of the creation of sub-islands caused by multiple power line trips.
  \end{rmk}

  \section{Simulation results}
  \label{sec:simulations}
  In this Section, we study the performance of the proposed PnP
  controllers. We consider the ImG in Figure \ref{fig:10DGUs_simulation},
  composed of 10 DGUs. All DGUs feed RL loads, except DGU 2 which is connected to
  a three-phase six-pulse diode
  rectifier. We notice a loop in the network that complicates
  the voltage regulation. Furthermore, DGUs are non-identical and all the electrical parameters are
  similar to those of the 11-DGUs example in \cite{Riverso_TSG}. The simulation (conducted in PSCAD) starts with DGUs 1-9 connected together and
  equipped with PnP controllers $\mathcal{C}_{[i]}$, $i = 1, \dots, 9$.
  \begin{figure}[!htb]
  	\centering
  	\begin{subfigure}[htb]{0.5\textwidth}
  		\centering
  		\begin{tikzpicture}[scale=.5]
	\ctikzset{bipoles/length=.4cm}
	\tikzstyle{every node}=[font=\sffamily\tiny,  minimum size=.7cm, inner sep=.7]

		\draw (1,1) node(6)  [circle, draw=blue, fill=blue!20] {DG 6};
		\draw (4,1) node(4)  [circle, draw=blue, fill=blue!20] {DG 4};
		\draw (7,1) node(2)  [circle, draw=blue, fill=blue!20] {DG 2};
		\draw (10,1) node(3)  [circle, draw=blue, fill=blue!20] {DG 3};
		\draw (13,1) node(7)  [circle, draw=blue, fill=blue!20] {DG 7};
		\draw (7,-2) node(10)  [circle, draw=blue, fill=blue!20] {DG 10};
		\draw (13,-2) node(8)  [circle, draw=blue, fill=blue!20] {DG 8};
		\draw (16,-2) node(9)  [circle, draw=blue, fill=blue!20] {DG 9};
		\draw (7,4) node(1)  [circle, draw=blue, fill=blue!20] {DG 1};
		\draw (4,4) node(5)  [circle, draw=blue, fill=blue!20] {DG 5};

		\draw[latex-latex] (4) to (6);
		
		\draw[latex-latex] (2) to (4);
		
		\draw[latex-latex] (3) to (2);
		
		\draw[latex-latex] (7) to (3);
		
		\draw[latex-latex] (8) to (7);
		
		\draw[latex-latex] (9) to (8);
		
		\draw[latex-latex] (4) to (6);
		
		\draw[latex-latex] (5) to (4);
		
		\draw[latex-latex] (1) to (5);		

		\draw[latex-latex] (1) to (2);
		
		\draw[red, dashed, latex-latex] (2) to (10);

		\draw[red, dashed, latex-latex] (10) to (8);
		
\end{tikzpicture}
  		\caption{ImG topologies until $t = 12$ s. For $0\leq t< 7.5$ s, DGUs 1-9 are interconnected (in
  			black); at $t = 7.5$ s, DGU 10 joins the network (in red).}
  		\label{fig:10DGUs_simulation_1}
  	\end{subfigure}
  	\begin{subfigure}[htb]{0.5\textwidth}
  		\centering
  		\begin{tikzpicture}[scale=.5]
\ctikzset{bipoles/length=.4cm}
\tikzstyle{every node}=[font=\sffamily\tiny,  minimum size=.7cm, inner sep=.7]

\draw (1,1) node(6)  [circle, draw=blue, fill=blue!20] {DG 6};
\draw (4,1) node(4)  [circle, draw=blue, fill=blue!20] {DG 4};
\draw (7,1) node(2)  [circle, draw=blue, fill=blue!20] {DG 2};
\draw (10,1) node(3)  [circle, draw=blue, fill=blue!20] {DG 3};
\draw (13,1) node(7)  [circle, draw=blue, fill=blue!20] {DG 7};
\draw (7,-2) node(10)  [circle, draw=blue, fill=blue!20] {DG 10};
\draw (13,-2) node(8)  [circle, draw=blue, fill=blue!20] {DG 8};
\draw (16,-2) node(9)  [circle, draw=blue, fill=blue!20] {DG 9};
\draw (7,4) node(1)  [circle, draw=blue, fill=blue!20] {DG 1};
\draw (4,4) node(5)  [circle, draw=blue, fill=blue!20] {DG 5};

\draw[latex-latex] (4) to (6);

\draw[latex-latex] (2) to (4);

\draw[latex-latex] (3) to (2);

\draw[latex-latex] (8) to (7);

\draw[latex-latex] (9) to (8);

\draw[latex-latex] (4) to (6);

\draw[latex-latex] (5) to (4);

\draw[latex-latex] (1) to (5);		

\draw[latex-latex] (1) to (2);

\draw[latex-latex] (2) to (10);

	\draw[blue, dashed] (0,-3) -- (11,-3) -- (11,5) -- (0,5)node[sloped, midway, above]{{ \small{$\mathcal{G}_{el}^{s1}$}}}  -- (0,-3);
	\draw[blue, dashed] (12,-3) -- (17,-3) -- (17,2) -- (12,2)node[sloped, midway, above]{{ \small{$\mathcal{G}_{el}^{s2}$}}}  -- (12,-3);
\end{tikzpicture}
  		\caption{Independent ImGs after the trip of lines 3-7 and 8-10 occurred at time $t = 12$ s.}
  		\label{fig:10DGUs_simulation_2}
  	\end{subfigure}
  	\caption{Simulation of a 10-DGUs ImGs: considered network topologies.}
  	\label{fig:10DGUs_simulation}
  \end{figure}
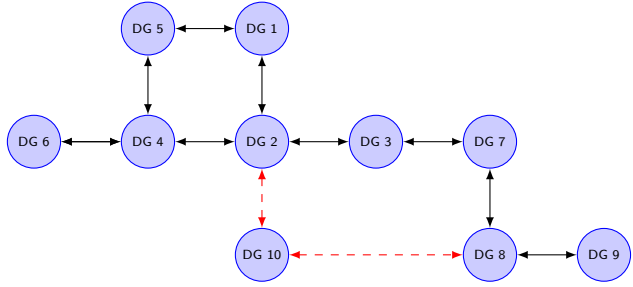
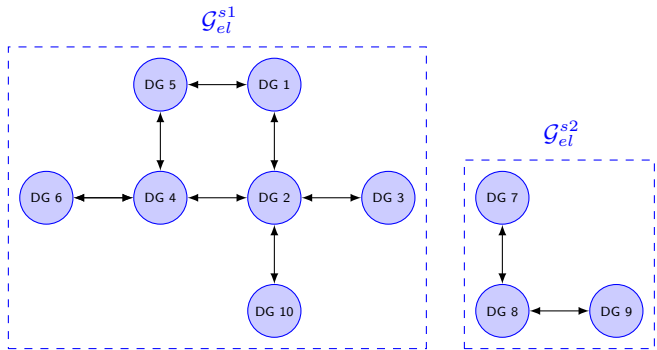
  
  As a first test, we validate the capability of PnP regulators to deal
  with real-time plugging in of DGUs. Therefore, at time $t = 7.5$ s, we
  simulate the connection of DGU 10, with DGUs 2 and 8 (see Figure
  \ref{fig:10DGUs_simulation_1}). 
  The $dq$ component of the voltages at PCCs 2, 8, 10 are shown in
  Figures \ref{fig:Vdq2}, \ref{fig:Vdq8} and \ref{fig:Vdq10}, respectively. Notably, we notice very small
  deviations of the DGUs voltages from their respective reference
  signals ($V^d_{2,ref}=0.6$ pu, $V^q_{2,ref}=0.5$ pu,
  $V^d_{8,ref}=0.7$ pu, $V^q_{8,ref}=0.6$ pu, and $V^q_{10,ref}=0.8$ pu, $V^q_{10,ref}=0.6$ pu). Furthermore, these deviations are
  compensated after a short transient. We also notice that the oscillations of $V^d_{2}$ and $V^q_{2}$ around their respective references are due to the fact that the corresponding signals in the $abc$ reference frame are not perfectly sinusoidal. This behavior is expected as nonlinear loads (like, e.g., the rectifier connected at PCC 2) inherently induce harmonic distortions on electrical variables.
  
  Next, in order to assess the robustness of the PnP-controlled ImG to
  load dynamics, at time $t= 10$ s we simulate an abrupt switch of the
  $RL$ load at PCC 10 (i.e. from $R = 60$ $\Omega$, $L = 0.02$ mH to $R = 120$ $\Omega$, $L
  = 0.02$ mH). From Figures \ref{fig:Vdq2}, \ref{fig:Vdq8} and \ref{fig:Vdq10}, we notice that the $d$ and $q$
  components of the voltages at PCCs 2, 8 and 10, do not significantly
  deviate from their references, thus revealing that step changes in
  the loads can be rapidly absorbed. Figure \ref{fig:freq} shows that the real-time plug-in of DGU
  10 and the load change at its PCC produce minor effects also on the
  frequency profiles of the PCC voltages. Notably, PnP regulators are capable to promptly restore the frequencies to
  the reference value (50 Hz) guaranteeing, overall, variations smaller than 0.6 Hz. In a similar way, we do not notice significant
  deviations from the reference RMS voltages (see Figure
  \ref{fig:Vrms}).
  
  Hereafter, we test the capabilities of PnP regulators to preserve voltage and frequency stability even when a sudden disconnection of a portion of the network occurs. To this purpose, at time $t =12$ s, we simulate the simultaneous trip of lines 3-7 and 8-10, leading to the formation of two independent ImGs (see the corresponding subgraphs $\mathcal{G}_{el}^{s1}$ and $\mathcal{G}_{el}^{s2}$ Figure \ref{fig:10DGUs_simulation_2}). In the light of Remark \ref{rmk:connected_graph}, the stability of the two networks is preserved (as shown in Figures \ref{fig:Vdq3}-\ref{fig:Vrms}), without the need to redesign any local controller. This feature can be useful in those scenarios where the disconnection of DGUs might need to be performed abruptly, due, for instance, to faults in the power lines.
  
  Finally, we notice that
  the Total Harmonic
  Distortion (THD) of the voltage at PCC 2 (whose local load is
  nonlinear) always remains below 5$\%$, which is the maximum limit
  recommended by IEEE standards \cite{IEEE2009} (see Figure \ref{fig:THD2}).
  
  \subsection{Robustness of PnP regulators to losses of clock synchronization}
  \label{sec:clock_desynch}
  In the following, we show that collective ImG voltage and frequency stability is not affected by losses of clock synchronization between DGUs equipped with local PnP regulators. More in details, 
  relative drifts exhibited by local oscillators (i.e. the clocks used to perform the $abc \rightarrow dq$ and $dq \rightarrow abc$ transformations in the gray boxes in Figure \ref{fig:ctrl_part}) will only have an impact on the tracking performance and not on collective stability. This is due to the fact that clock desynchronization acts as an additive disturbance on the $dq$ side of the $abc-dq$ transformation blocks. Consequently, these drifts have the same effect of disturbances on the local controller inputs and outputs (the blue and red arrows in Figure \ref{fig:ctrl_part}). 
  From linear system theory, we expect such exogenous disturbances only to affect performance, without compromising the asymptotic stability of the ImG. To support this point, we have run the same simulation described in Section \ref{sec:simulations}, adding now significant 
  phase shifts\footnote{The implemented phase shifts are much higher than the worst case ones induced by currently available crystal oscillators (whose error $e$ is comprised between $2\times 10^{-6}$ and $ 2\times 10^{-11}$ seconds per year \cite{Guerrero2013,Etemadi2012a}). Assuming $\omega_0 = 2\pi 50$ [rad/s] and $e=2\times 10^{-6}$ [s/year], after one year, the worst case phase shift is $\tilde{\theta}_{0, \mathrm{DGU}_i} = e_{\mathrm{DGU}_i}\cdot\omega_0 \cdot \frac{180}{\pi} = \left(3.6 \times 10^{-2}\right)^\circ$.}  
  $\tilde{\theta}_{0,\mathrm{DGU}_i}$ in the clocks of DGUs 2, 3, 7, 8 and 10 ($\tilde{\theta}_{0,\mathrm{DGU}_2} = 2.5^\circ$, $\tilde{\theta}_{0,\mathrm{DGU}_3} = 1^\circ$, $\tilde{\theta}_{0,\mathrm{DGU}_7} = 2^\circ$, $\tilde{\theta}_{0,\mathrm{DGU}_7} = 2^\circ$, $\tilde{\theta}_{0,\mathrm{DGU}_8} = 3^\circ$, $\tilde{\theta}_{0,\mathrm{DGU}_{10}} = 4^\circ$). The subplots in Figure \ref{fig:overall_sim_nosynch} (to be compared with the corresponding ones in Figure \ref{fig:overall_sim})
  reveal that, in this new scenario, the tracking of reference signals given in Section \ref{sec:simulations} is affected by offsets.
  Nonetheless, asymptotic stability is always preserved in spite of changes in the ImG topology and in the load condition.
  \begin{rmk}
  	For the sake of completeness, we highlight that, besides relying on high precision local oscillators, clock synchronization 		
  	can be achieved through GPS radio clock with high accuracy \cite{4579760} or via packet networks, using either distributed protocols \cite{gusella1989accuracy} or approaches based on all-to-all communication \cite{IEEE2017}.
  \end{rmk}
  \begin{figure*}[!htb]
  	\centering
  	\begin{subfigure}[!htb]{0.23\textwidth}
  		\centering
  		\includegraphics[width=1.15\textwidth]{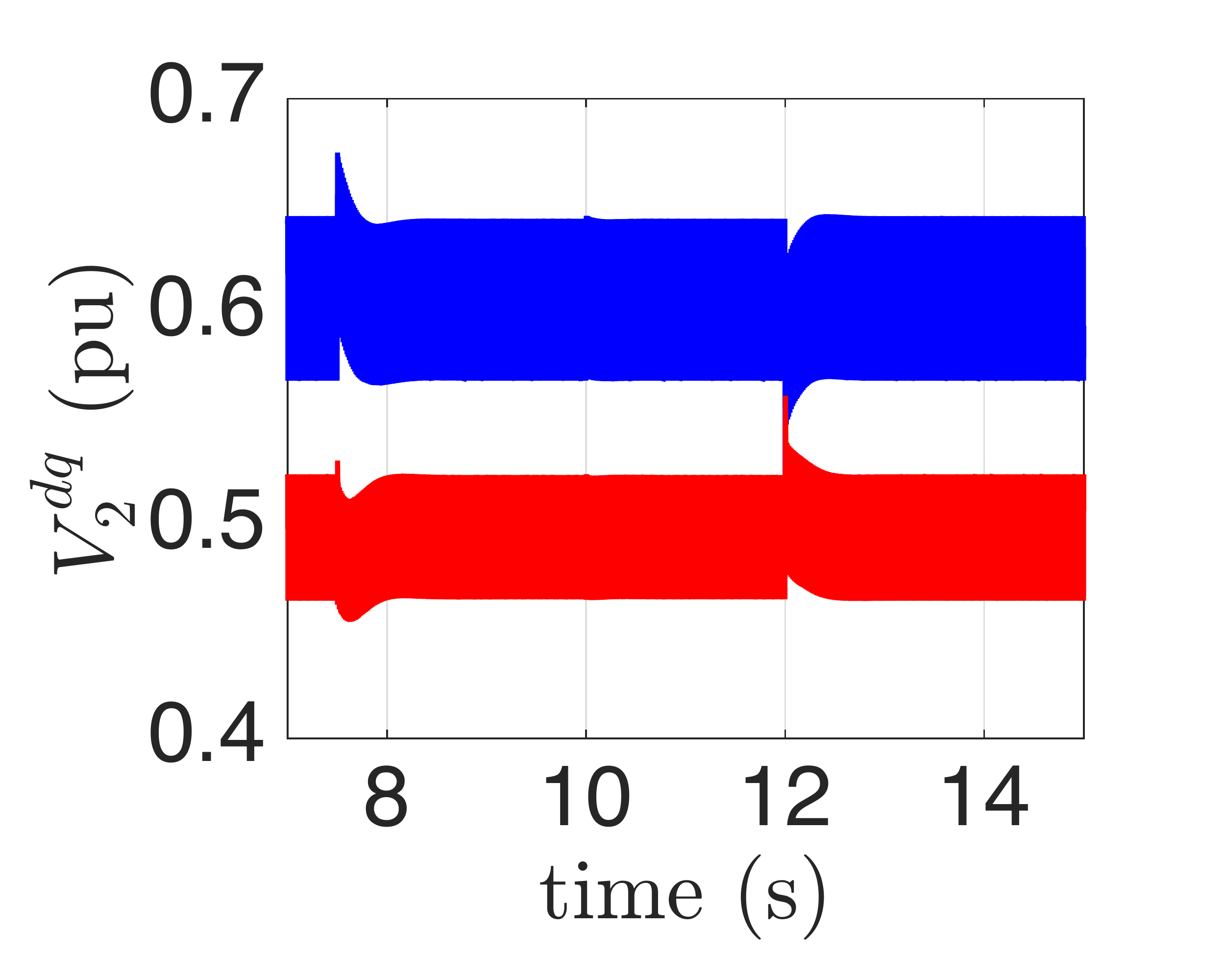}
  		\caption{\emph{d} (blue) and \emph{q} (red) components of the voltage at $PCC_2$.}
  		\label{fig:Vdq2}
  	\end{subfigure}\hspace{2mm}
  	\begin{subfigure}[!htb]{0.23\textwidth}
  		\centering
  		\includegraphics[width=1.15\textwidth]{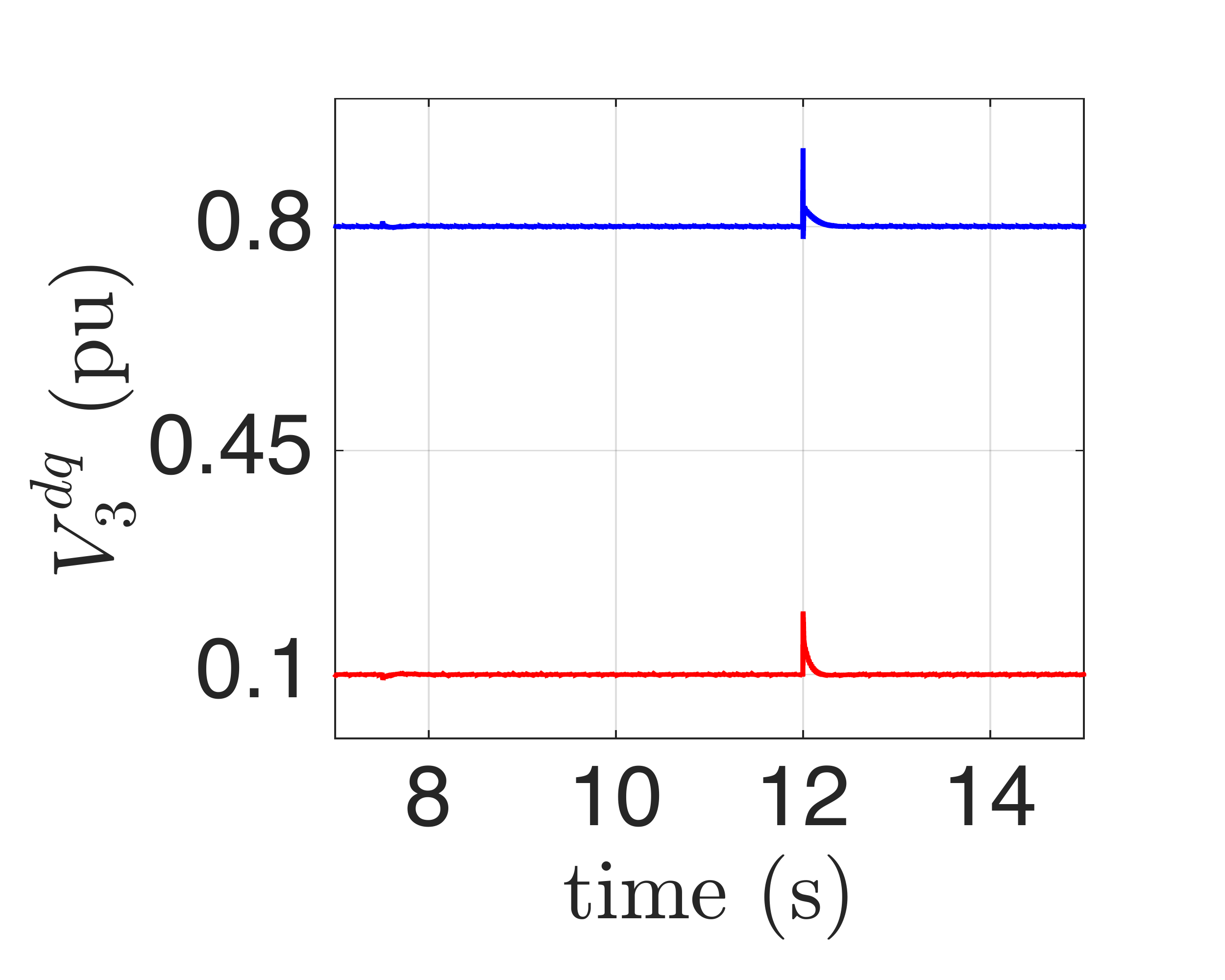}
  		\caption{\emph{d} (blue) and \emph{q} (red) components of the voltage at $PCC_3$.}
  		\label{fig:Vdq3}
  	\end{subfigure}\hspace{2mm}
  	\begin{subfigure}[!htb]{0.23\textwidth}
  		\centering
  		\includegraphics[width=1.15\textwidth]{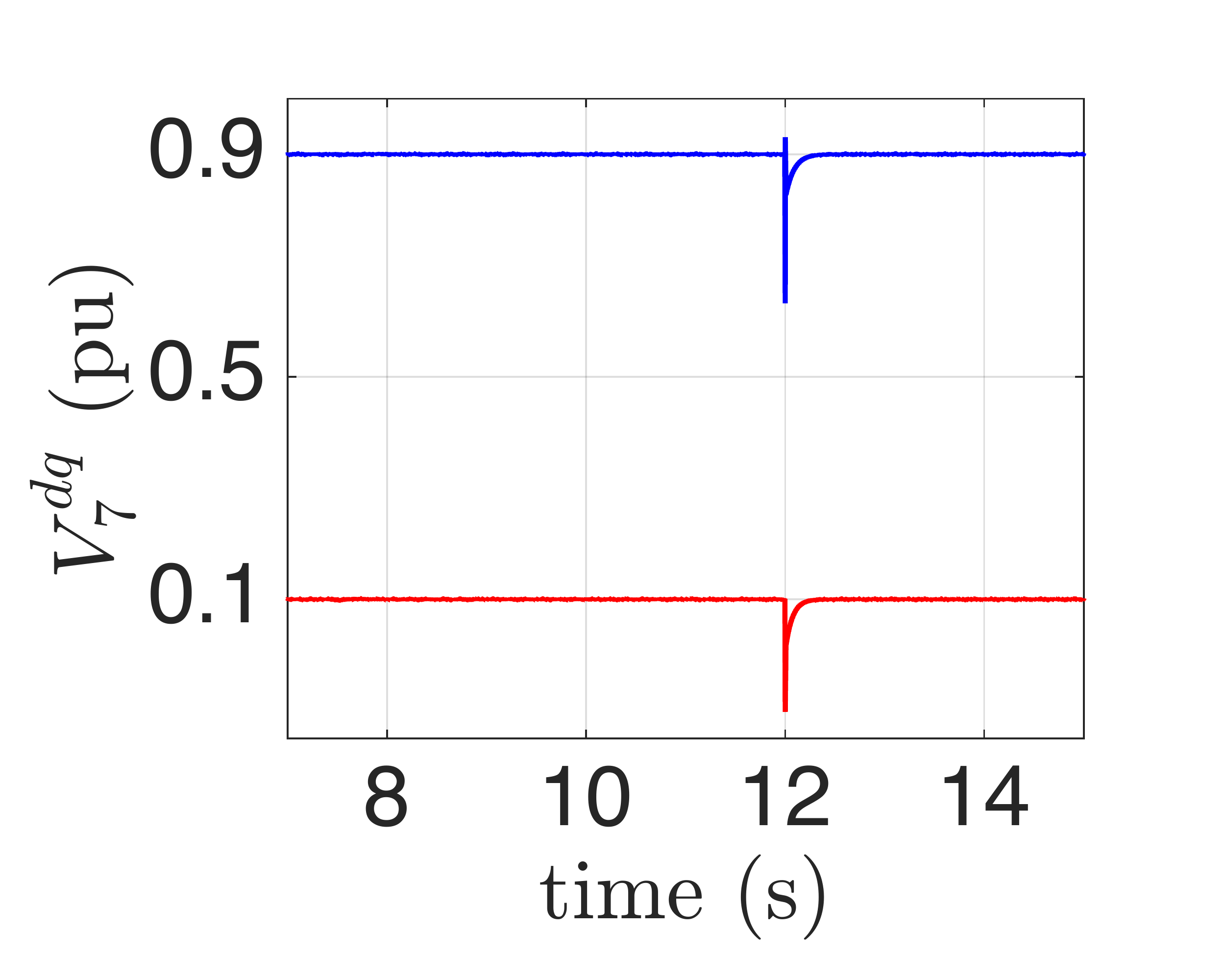}
  		\caption{\emph{d} (blue) and \emph{q} (red) components of the voltage at $PCC_7$.}
  		\label{fig:Vdq7}
  	\end{subfigure}\hspace{2mm}
  	\begin{subfigure}[!htb]{0.23\textwidth}
  		\centering
  		\includegraphics[width=1.15\textwidth]{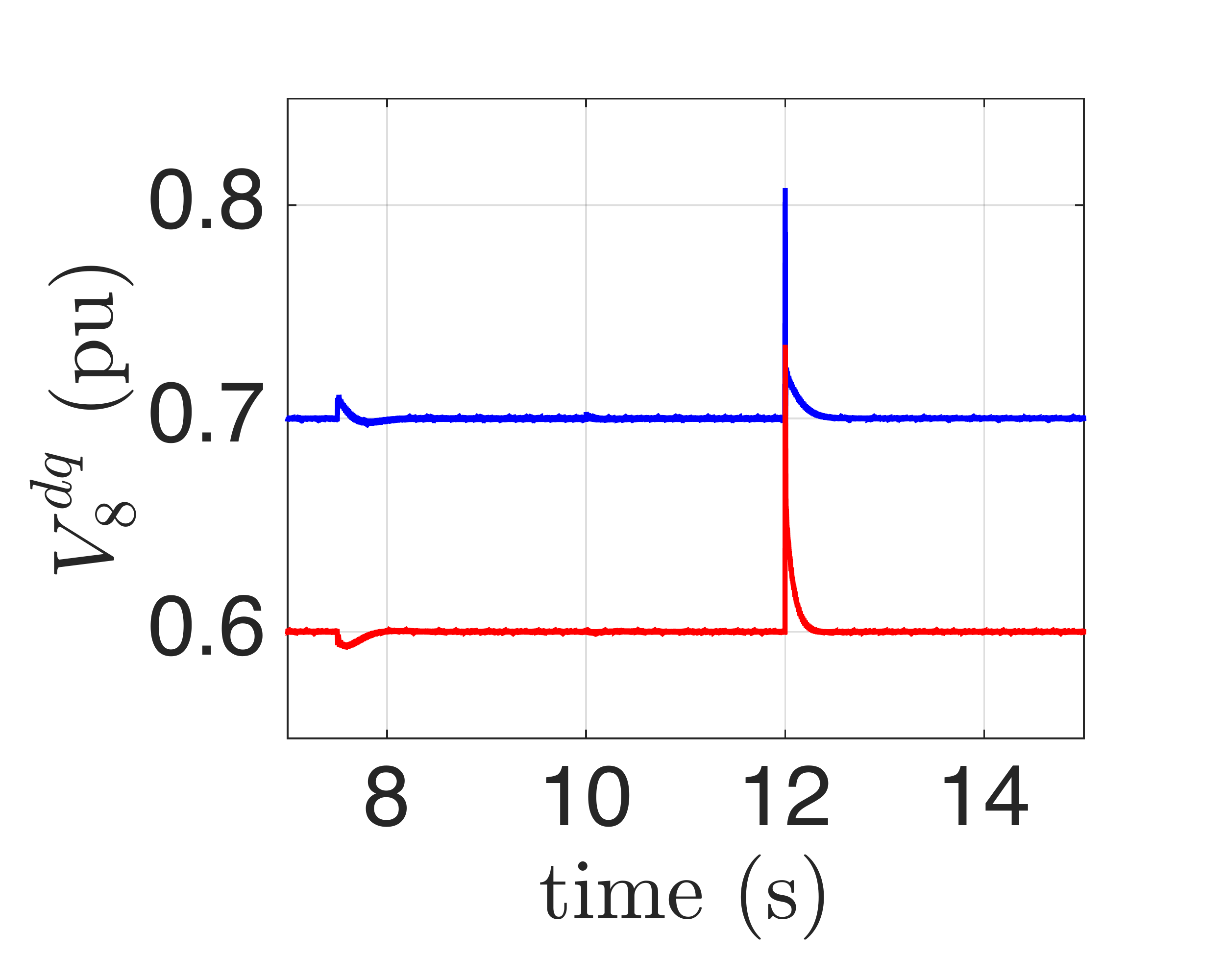}
  		\caption{\emph{d} (blue) and \emph{q} (red) components of the voltage at $PCC_8$.}
  		\label{fig:Vdq8}
  	\end{subfigure}
  	\begin{subfigure}[!htb]{0.23\textwidth}
  		\centering
  		\includegraphics[width=1.15\textwidth]{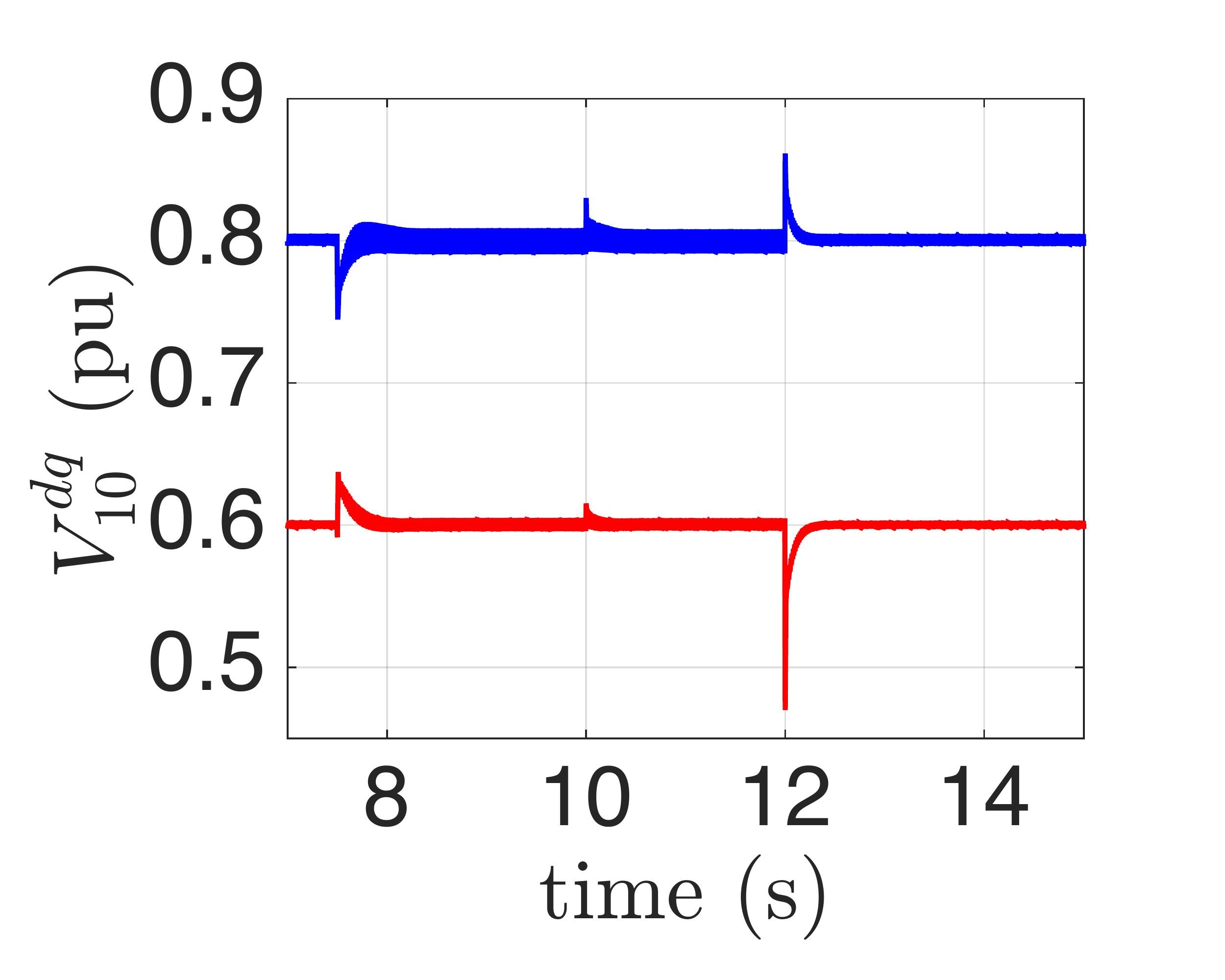}
  		\caption{\emph{d} (blue) and \emph{q} (red) components of the voltage at $PCC_{10}$.}
  		\label{fig:Vdq10}
  	\end{subfigure}\hspace{2mm}
  	\begin{subfigure}[!htb]{0.23\textwidth}
  		\centering
  		\includegraphics[width=1.15\textwidth]{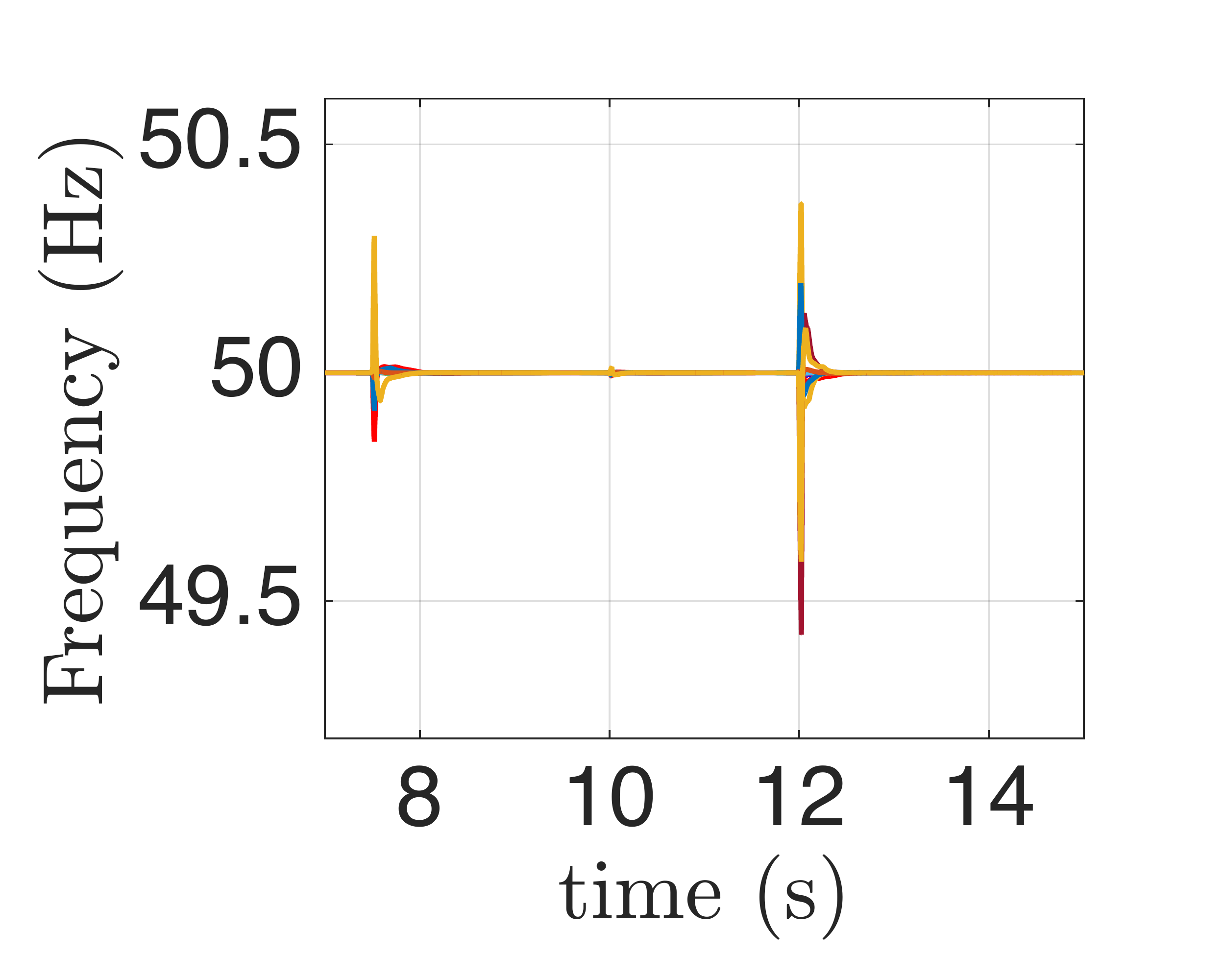}
  		\caption{Frequency of phase $a$ at the $PCC$s.}
  		\label{fig:freq}
  	\end{subfigure}\hspace{2mm}
  	\begin{subfigure}[!htb]{0.23\textwidth}
  		\centering
  		\includegraphics[width=1.15\textwidth]{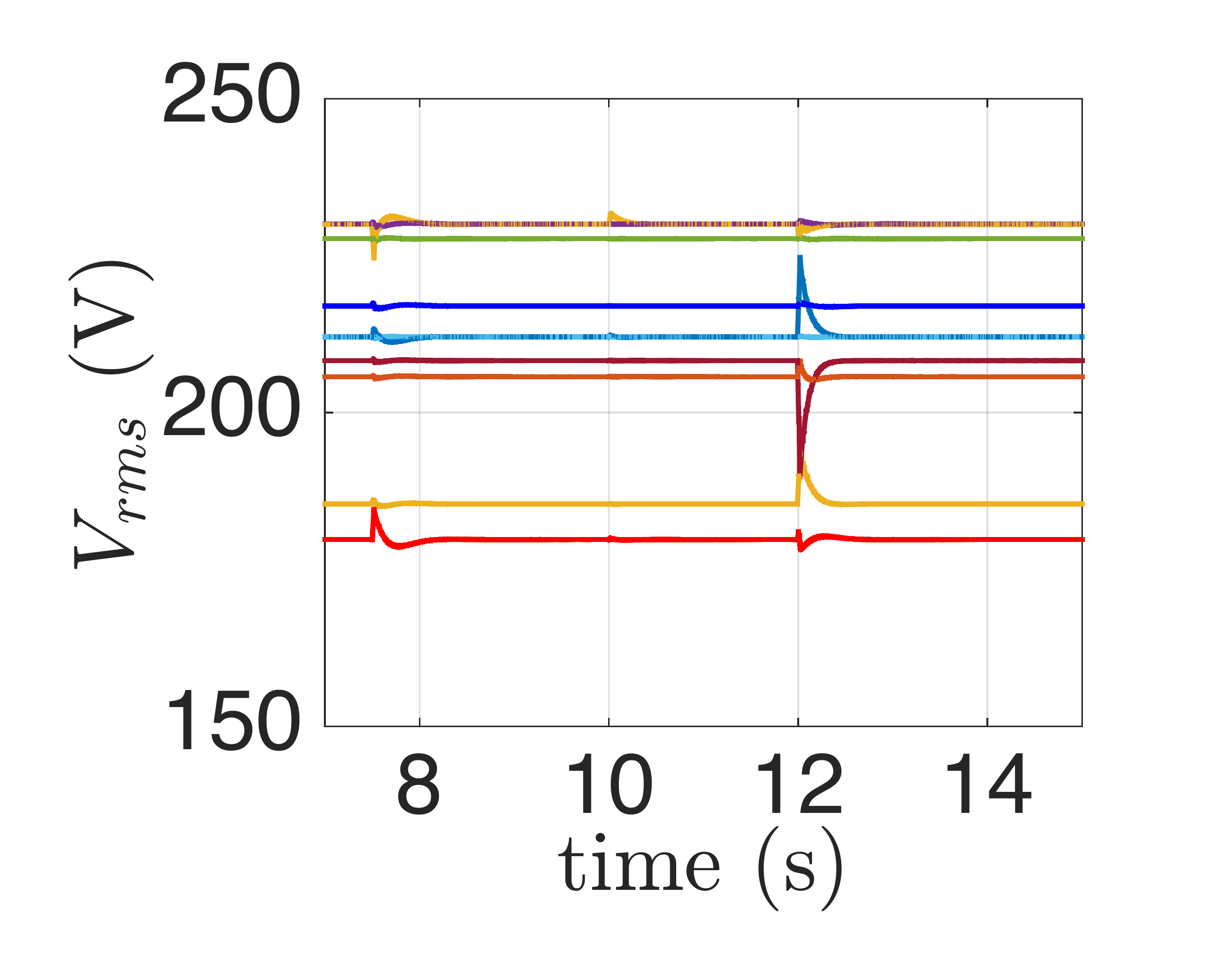}
  		\caption{RMS voltages of phase $a$ at the $PCC$s.}
  		\label{fig:Vrms}
  	\end{subfigure}\hspace{2mm}
  	\begin{subfigure}[!htb]{0.23\textwidth} 
  		\centering
  		\includegraphics[width=1.15\textwidth]{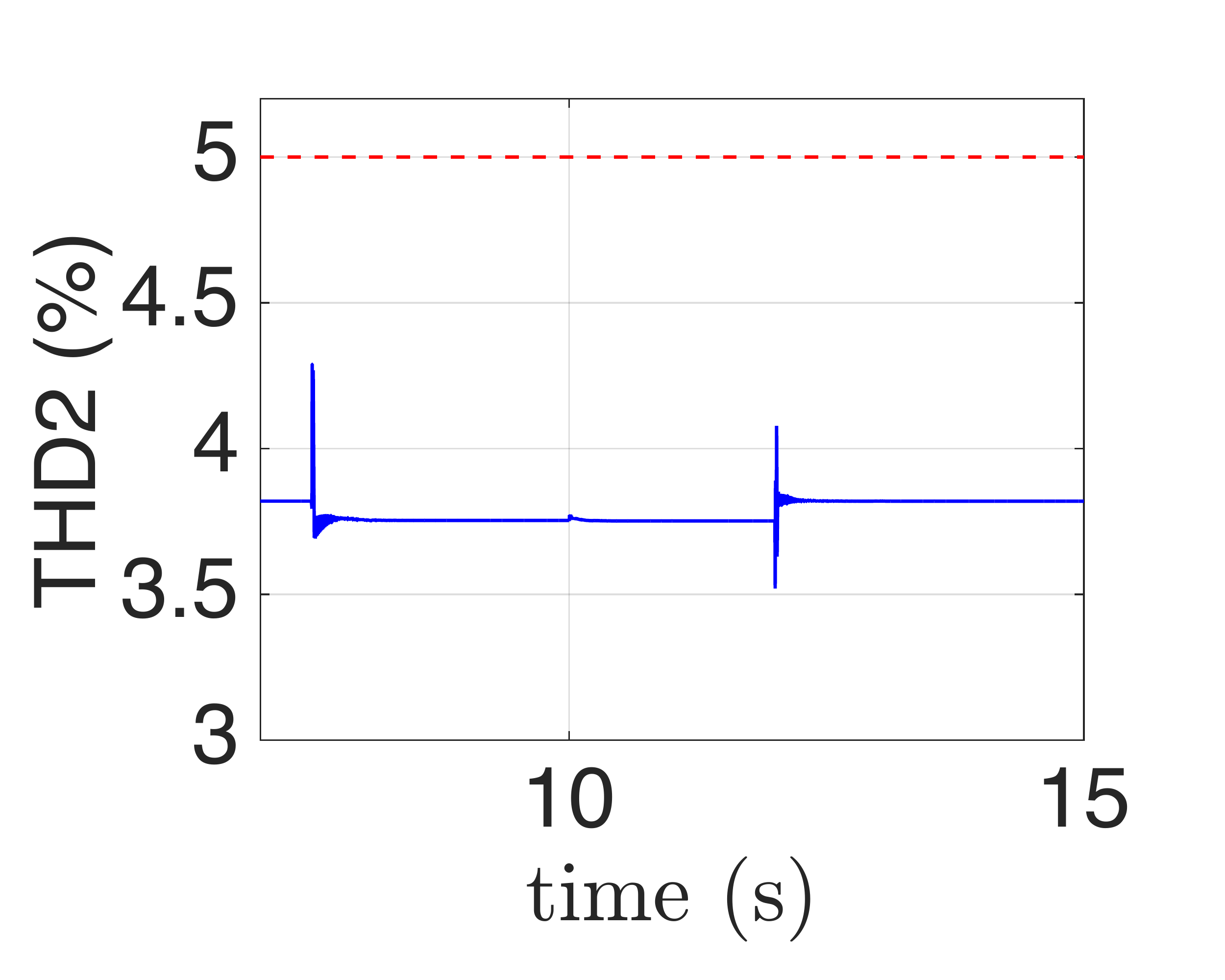}
  		\caption{THD of phase $a$ of the voltage at $PCC_2$.}
  		\label{fig:THD2}
  	\end{subfigure}
  	\caption{Performance of PnP voltage and
  		frequency control. Connection of DGU 10, load change at PCC 10, and disconnection of DGUs 3 and 7, are performed at times $t =7.5$ s, $t = 10$ s and $t = 12$ s, respectively.}
  	\label{fig:overall_sim}   
  \end{figure*}
  
  \begin{figure*}[!htb]
  	\centering
  	\begin{subfigure}[!htb]{0.23\textwidth}
  		\centering
  		\includegraphics[width=1.15\textwidth]{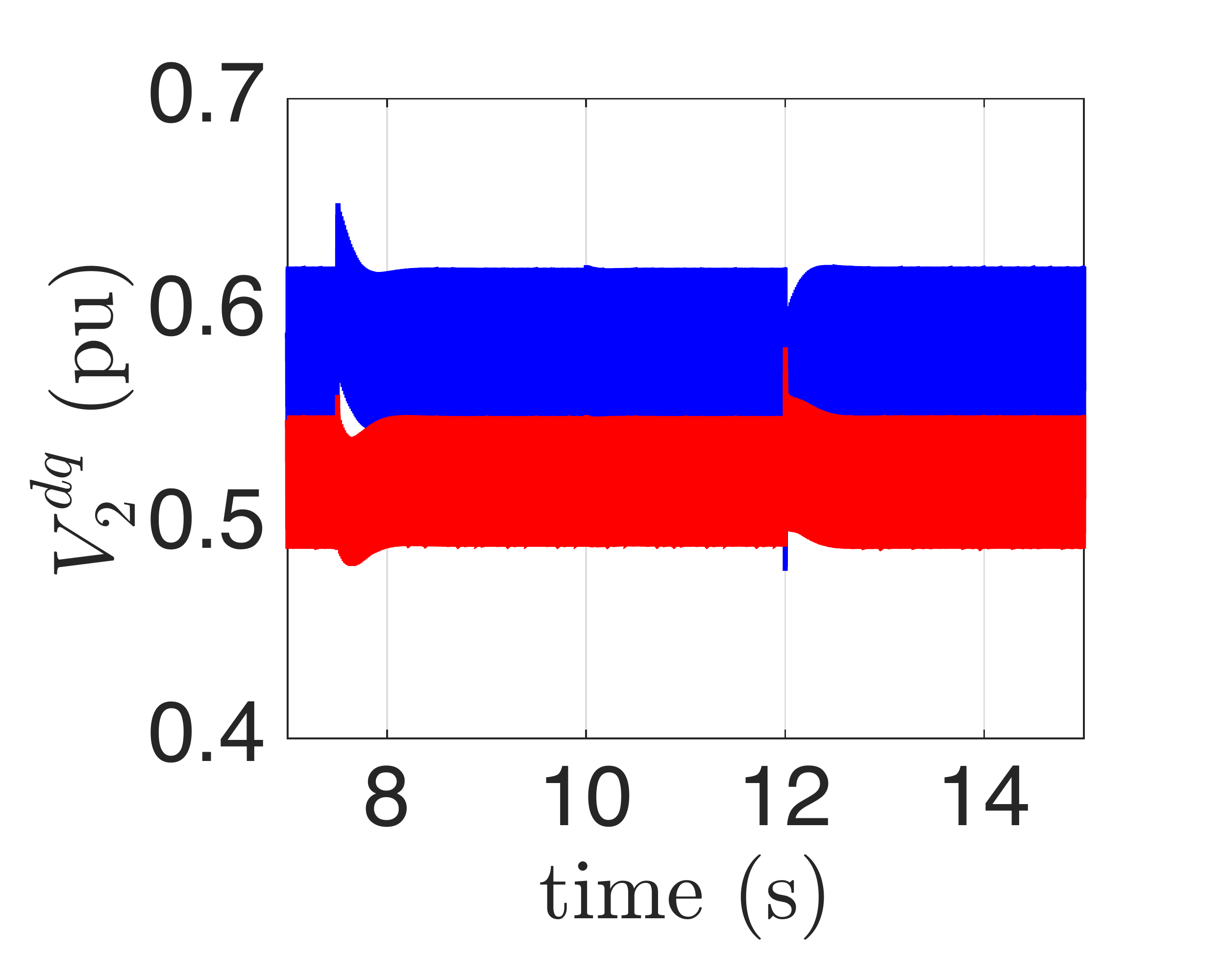}
  		\caption{\emph{d} (blue) and \emph{q} (red) components of the voltage at $PCC_2$.}
  		\label{fig:Vdq2ns}
  	\end{subfigure}\hspace{2mm}
  	\begin{subfigure}[!htb]{0.23\textwidth}
  		\centering
  		\includegraphics[width=1.15\textwidth]{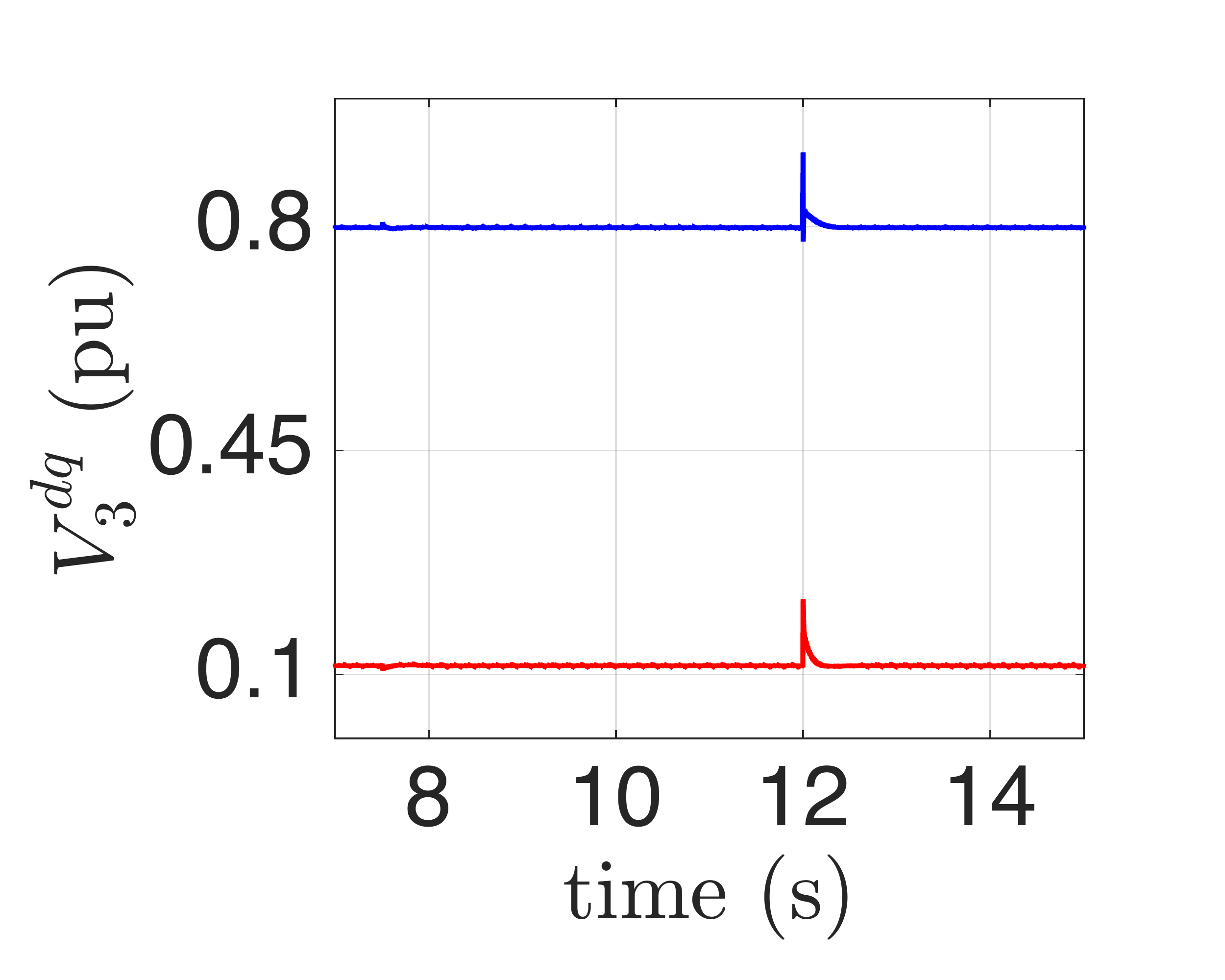}
  		\caption{\emph{d} (blue) and \emph{q} (red) components of the voltage at $PCC_3$.}
  		\label{fig:Vdq3ns}
  	\end{subfigure}\hspace{2mm}
  	\begin{subfigure}[!htb]{0.23\textwidth}
  		\centering
  		\includegraphics[width=1.15\textwidth]{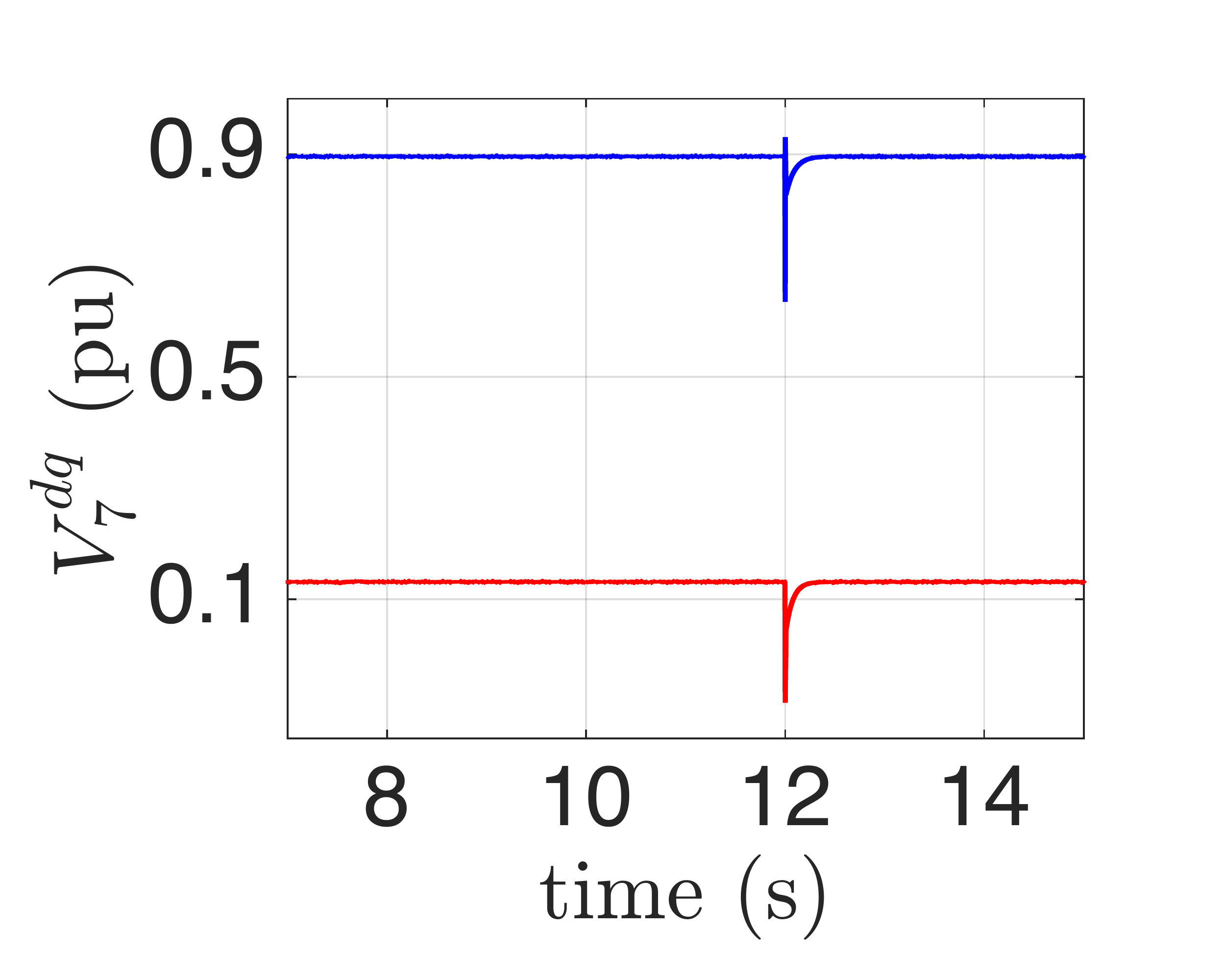}
  		\caption{\emph{d} (blue) and \emph{q} (red) components of the voltage at $PCC_7$.}
  		\label{fig:Vdq7ns}
  	\end{subfigure}\hspace{2mm}
  	\begin{subfigure}[!htb]{0.23\textwidth}
  		\centering
  		\includegraphics[width=1.15\textwidth]{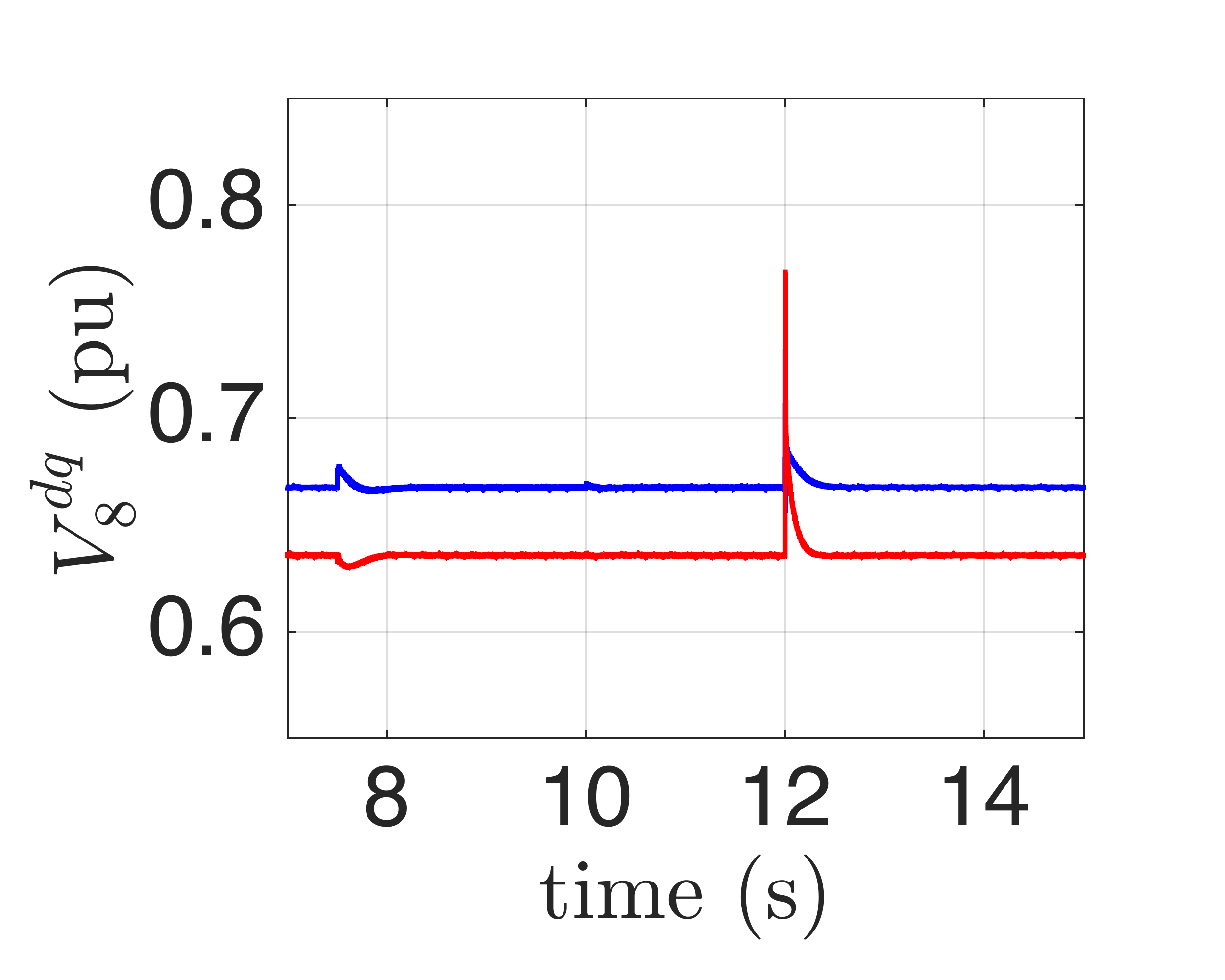}
  		\caption{\emph{d} (blue) and \emph{q} (red) components of the voltage at $PCC_8$.}
  		\label{fig:Vdq8ns}
  	\end{subfigure}
  	\begin{subfigure}[!htb]{0.23\textwidth}
  		\centering
  		\includegraphics[width=1.15\textwidth]{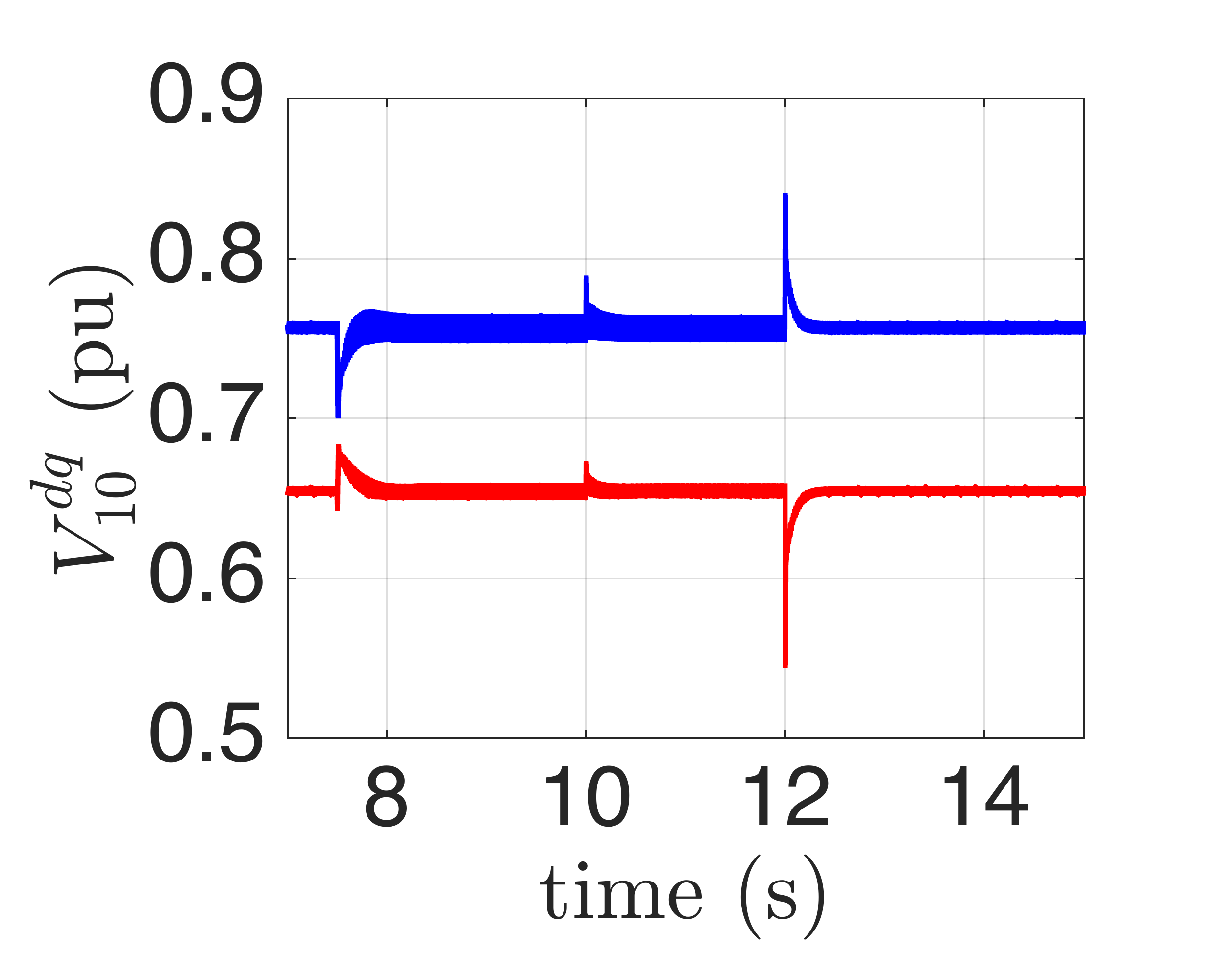}
  		\caption{\emph{d} (blue) and \emph{q} (red) components of the voltage at $PCC_{10}$.}
  		\label{fig:Vdq10ns}
  	\end{subfigure}\hspace{2mm}
  	\begin{subfigure}[!htb]{0.23\textwidth}
  		\centering
  		\includegraphics[width=1.15\textwidth]{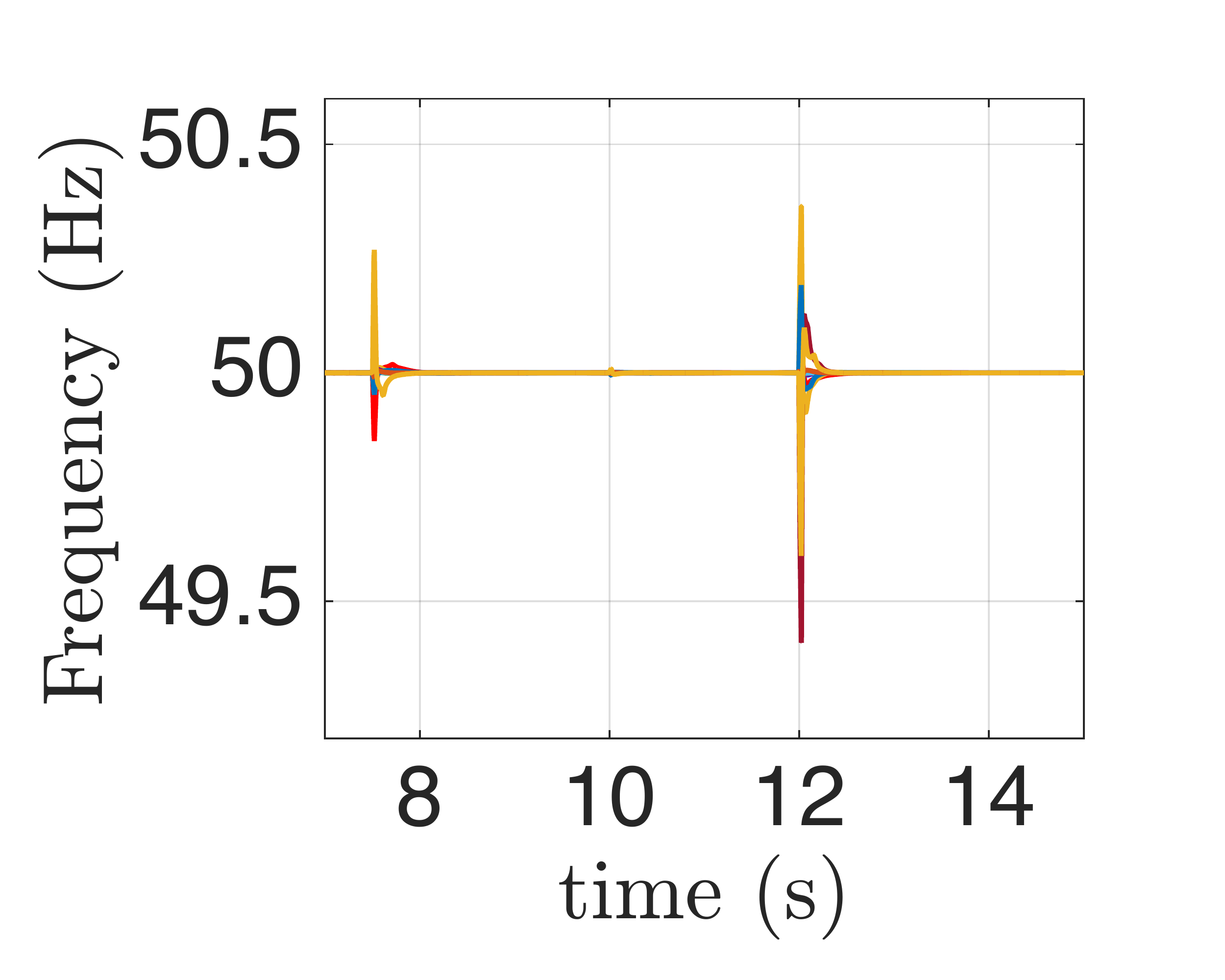}
  		\caption{Frequency of phase $a$ at the $PCC$s.}
  		\label{fig:freqns}
  	\end{subfigure}\hspace{2mm}
  	\begin{subfigure}[!htb]{0.23\textwidth}
  		\centering
  		\includegraphics[width=1.15\textwidth]{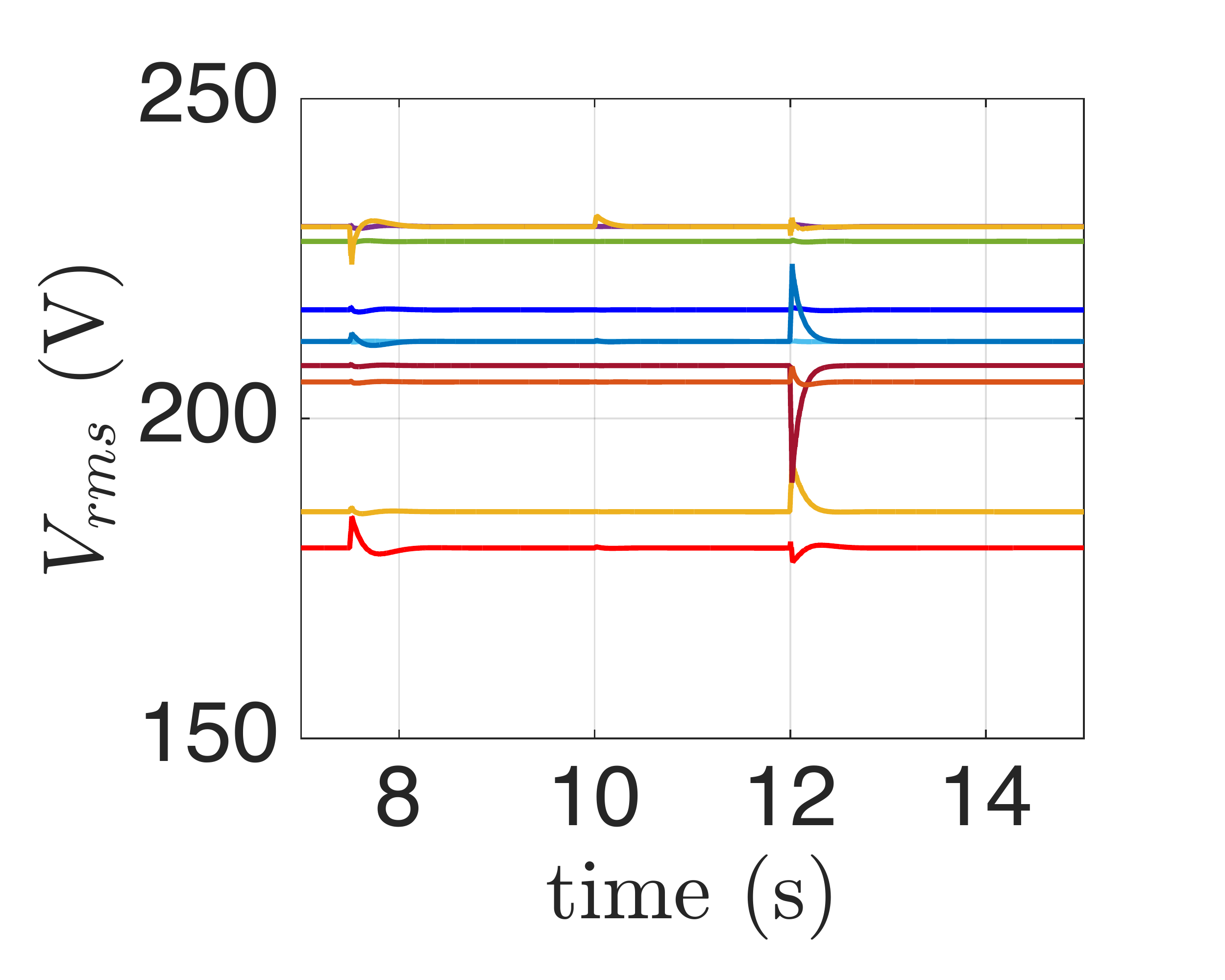}
  		\caption{RMS voltages of phase $a$ at the $PCC$s.}
  		\label{fig:Vrmsns}
  	\end{subfigure}\hspace{2mm}
  	\begin{subfigure}[!htb]{0.23\textwidth} 
  		\centering
  		\includegraphics[width=1.15\textwidth]{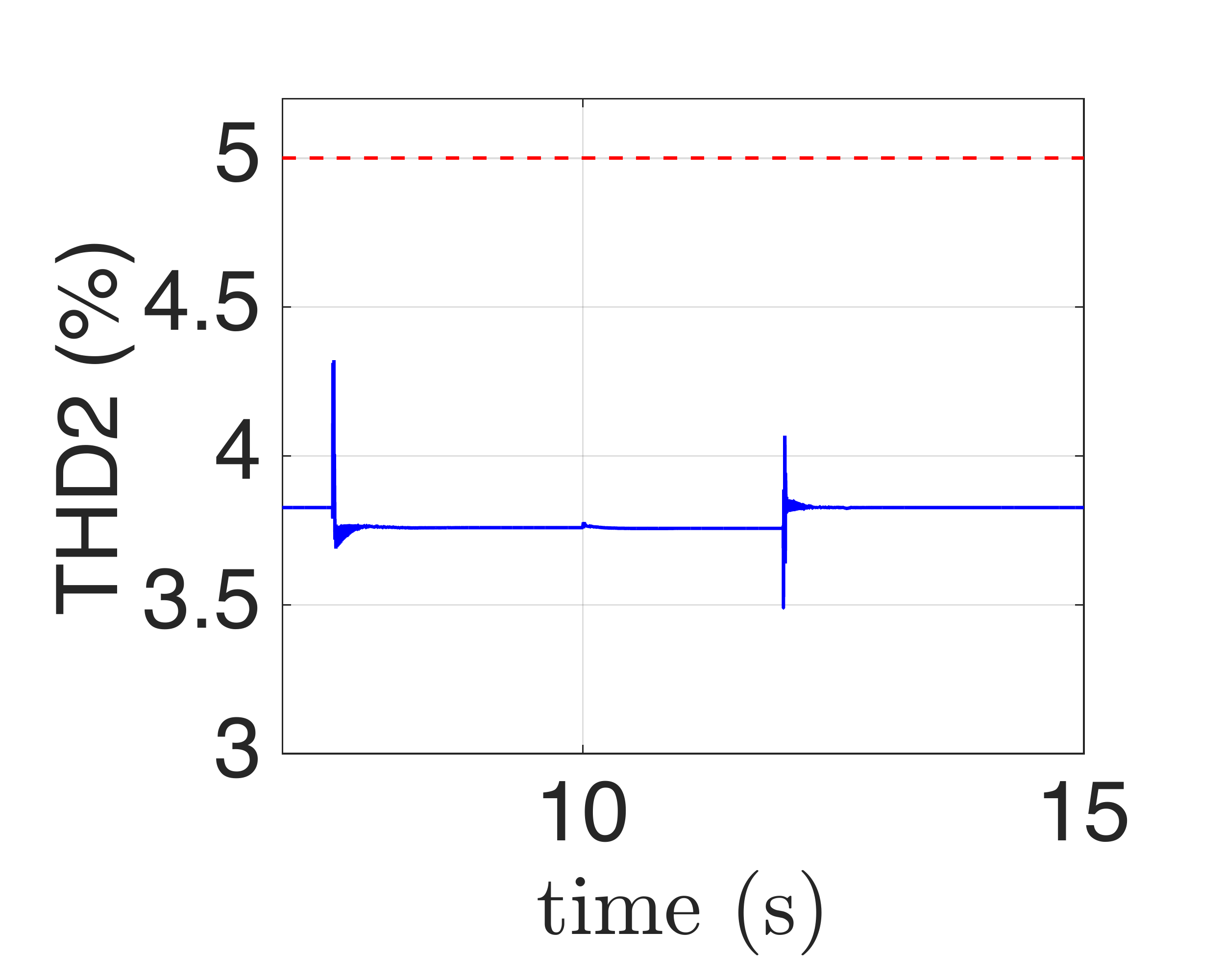}
  		\caption{THD of phase $a$ of the voltage at $PCC_2$.}
  		\label{fig:THD2ns}
  	\end{subfigure}
  	\caption{Performance of PnP voltage and
  		frequency control when significant phase shifts are simulated for the clocks of DGUs 2, 3, 7, 8 and 10. Connection of DGU 10, load change at PCC 10, and disconnection of DGUs 3 and 7, are performed at times $t =7.5$ s, $t = 10$ s and $t = 12$ s, respectively.}
  	\label{fig:overall_sim_nosynch}   
  \end{figure*}
  
  \section{Conclusions}
  In this paper, we presented a decentralized control approach to voltage
  and frequency stabilization in AC ImG. Differently from the PnP
  methodology in \cite{Riverso_TSG}, the presented procedure is always feasible and guarantees
  overall ImG stability while computing local
  controllers in a line-independent fashion. Future research will focus on studying how to couple
  PnP local regulators with a higher control layer for power flow regulation among DGUs.

  \appendix
  \section{Proof of Theorem \ref{thm:overall_stability}}
  Before showing the proof of Theorem \ref{thm:overall_stability},
  in the next  Proposition we provide preliminary results which are instrumental for
  characterizing the states $\mathbf{\hat x}$ yielding $\dot{\mathcal{V}}(\mathbf{\hat x})=0$.

  \begin{prp}
  	\label{pr:quadr_form}
  	Let Assumption \ref{ass:ctrl} hold, \eqref{eq:new_Lyapunov} be verified, $Y_i>0$, and $\tilde Q_i$ and $Q_i$ be given by 
  	\eqref{eq:Qi_tilde_struct} and \eqref{eq:Qtilde}, respectively. Moreover, consider
  	$g_i(w_i) = w_i^T Q_i w_i$, with $w_i\in\mathbb{R}^6$. Only vectors $\bar w_i$ in the form
  	\begin{equation*}
  	\label{eqn:prop2}
  	\bar{w}_i =\left[ \begin{array}{c|c|c}
  	\eta_i^{-1}\alpha_i^T&
  	\bar \sigma^{-1} \beta_i^T  &
  	(\mathcal{Y}_{33,i}\beta_i)^T
  	\end{array}\right]^T
  	\end{equation*}
  	with $\alpha_i , \beta_i\in\mathbb{R}^2$ fulfill
  	\begin{equation}
  	\label{eq:vQv=0}
  	g_i(\bar{w}_i) =0.
  	\end{equation}
  \end{prp}
  \begin{proof}
  	For the sake of simplicity, hereafter we omit the subscript
  	$i$. 
  	The assumptions imply that $Q$ is
  	negative semidefinite. Then, vectors $\bar w$ satisfying \eqref{eq:vQv=0} also maximize
  	$g(\cdot)$, hence verifying 
  	\begin{equation}
  	\label{eq:maximum}
  	\frac{\mathrm{d}g}{\mathrm{d}w}(\bar w)
  	= 2Q\bar w=\mathbf{0}.
  	\end{equation}
  	%
  	From \eqref{eq:Qtilde} we have 
  	\begin{equation}
  	\label{eq:Qtilde_LaSalle}
  	Q = Y^{-1}\tilde{Q}Y^{-1}.
  	\end{equation}
  	By replacing \eqref{eq:Qtilde_LaSalle} in \eqref{eq:maximum} and pre-multiplying the obtained expression by $Y$, we get
  	\begin{equation}
  	\label{eq:maximum_tilde}
  	\tilde Q\underbrace{Y\bar w}_{\tilde w}=\mathbf{0}.
  	\end{equation} 
  	Set ${\tilde w} = \left[ \begin{array}{ccc}
  	\tilde{\alpha}^T  &	\tilde\beta^T & \tilde{\gamma}^T  	\end{array}\right]^T$, $\tilde{\alpha}$, $\tilde{\beta}$, $\tilde{\gamma}\in \Rset^2$. Proposition \ref{pr:Qtilde_semidef} can be applied and 
  	\eqref{eq:maximum_tilde} becomes \[
  	\tilde Q\tilde w = \left[ \begin{array}{ccc}
  	\mathbf{0}  &	(\tilde{\mathcal{Q}}_{22}\tilde\beta)^T & \mathbf{0} \end{array}\right]^T = \mathbf{0}. 
  	\] 
  	From \eqref{eq:new_Lyapunov}, $\tilde \QQ_{22}$ is nonsingular and
  	$\tilde{\mathcal{Q}}_{22}\tilde\beta = \mathbf{0}$ implies $\tilde\beta = \mathbf{0}$. Therefore, $\bar w = \left[ \begin{array}{ccc}
  	{\alpha}^T  &	\beta^T & {\gamma}^T  	\end{array}\right]^T$ verifies \eqref{eq:maximum} if and only if 
  	\[
  	Y^{-1}\left[ \begin{array}{ccc}{\alpha}^T &	\beta^T & {\gamma}^T \end{array}\right]^T = \left[ \begin{array}{ccc}
  	\tilde{\alpha}^T  &	\mathbf{0} & \tilde{\gamma}^T  	\end{array}\right]^T.
  	\]
  	By pre-multiplying both sides by $Y$ we obtain
  	\[
  	\left[\begin{array}{ccc}
  	\eta^{-1}I_2 & \mathbf{0}_2 & \mathbf{0}_2 \\
  	\mathbf{0}_2 & \mathcal{Y}_{22} & {\bar \sigma}^{-1}I_2\\
  	\mathbf{0}_2 & {\bar\sigma}^{-1}I_2 & \mathcal{Y}_{33} \\
  	\end{array}\right]\left[\begin{array}{c}
  	\tilde{\alpha}\\
  	\mathbf{0} \\
  	\tilde{\gamma}
  	\end{array}\right]  = \left[\begin{array}{c}
  	\eta^{-1}\tilde{\alpha}\\
  	{\bar\sigma}^{-1}\tilde{\gamma}\\
  	\mathcal{Y}_{33}\tilde{\gamma}
  	\end{array}\right].
  	\]
  \end{proof}
  
\subsection{Proof of Theorem \ref{thm:overall_stability}}
\begin{proof}
	Since system \eqref{eq:sysaugoverallclosed} is linear, we can neglect the input $\hat{\mathbf{d}}$ and focus on the stability of the origin. Any solution to Problem \ref{prbl:problem_3} verifies \eqref{eq:new_Lyapunov}, which implies that $\tilde Q_i$ in \eqref{eq:Qi_tilde_struct} is negative semidefinite. Moreover $Y_i>0$ and, from \eqref{eq:Qtilde}, the inequality \eqref{eq:Qi_semidef}  holds.
	
	From Proposition \ref{pr:semidefinite_abc}, we have that
	\eqref{eq:Lyapeqnoverall} is verified. Therefore, we aim to use the LaSalle
	invariance Theorem \cite{khalil2001nonlinear} to show that the origin
	of the ImG is attractive.
	Let us compute the set
	$R = \{\mathbf{x}\in\mathbb{R}^{6N} : \mathbf{x}^T \mathbf{Q}\mathbf{x}= 0 \}$, which,
	using \eqref{eq:Lyap_abc}, can be written as 
	\begin{equation}
	\label{eq:R}
	\begin{aligned}
	R &= \{\mathbf{x} : \mathbf{x}^T \left( (a)+(b)+(c)\right)\mathbf{x}
	= 0 \}\\
	&=\{\mathbf{x} : \mathbf{x}^T (a) \mathbf{x} +\mathbf{x}^T(b)\mathbf{x}+\mathbf{x}^T(c) \mathbf{x}
	= 0 \}\\
	&=\underbrace{\{\mathbf{x} : \mathbf{x}^T (a) \mathbf{x}
		=0\}}_{X_1}\cap
	\underbrace{\{\mathbf{x}:\mathbf{x}^T\left[(b)+(c)\right] \mathbf{x}
		=0\}}_{X_2}, 
	\end{aligned}
	\end{equation}
	and first focus on characterizing the vectors of set $X_1$. 
	Recall
	that $(a)=\text{diag}(Q_1, \dots, Q_N)$. From Problem \ref{prbl:problem_3}, all assumptions of Proposition \ref{pr:quadr_form} are verified and, by utilizing it,
	we have
	\begin{small}
		\begin{equation*}
		\begin{aligned}
		X_1 = \{\mathbf{x} : \mathbf{x} =&\left[  \text{ }\alpha_1^T\text{ }\beta_1^T
		\text{ } \gamma_1^T \text{ }|\text{ }\cdots \text{ }| \text{ }\alpha_N^T\text{ }\beta_N^T
		\text{ } \gamma_N^T \text{ }\right]^T, \\
		&\alpha_i,\beta_i,
		\gamma_i\in\mathbb{R}^2,
		\beta_i = {\bar \sigma}^{-1}\mathcal{Y}_{33,i}^{-1}\gamma_i\}.
		\end{aligned}
		\end{equation*}
	\end{small}
	Next, we focus on the elements of $X_2$. We have seen that
	the term $(b)+(c)$ is an "expansion" of the Laplacian matrix in 
	\eqref{eq:laplacian}, obtained by augmenting each $2\times
	2$ block $\Phi_{ij}$ of $\mathcal{L}$ with zero rows and
	columns, so as to retrieve blocks of dimension $6\times
	6$. It follows that, by construction, $X_2$ contains vectors in the form
	\begin{small}
		\begin{equation}
		\label{eq:X2_a}
		\mathbf{\tilde x} =\left[  \text{ }\mathbf{0}\text{ }\tilde x_{12}^T \text{ }\tilde
		x_{13}^T \text{ }|\text{ }\cdots \text{ }|\text{ }\mathbf{0}\text{ }\tilde
		x_{N2}^T \text{ }\tilde x_{N3}^T \text{ }\right]^T,\hspace{3mm}\tilde
		x_{i2}, \tilde x_{i3}\in\mathbb{R}^2,\forall i\in\mathcal{D}.
		\end{equation}
	\end{small}
	Moreover, since the kernel of the Laplacian matrix of a connected
	graph contains only vectors with identical entries
	\cite{godsil2001algebraic}, we also have that
	\begin{small}
		\begin{equation}
		\label{eq:X2_b}
		\{\mathbf{\bar{x}}=\left[  \text{ }\bar{x}^T\text{ }\mathbf{0} \text{ } \mathbf{0} \text{
		}|\text{ }\cdots \text{ }|\text{ }\bar{x}^T\text{ } \mathbf{0} \text{ } \mathbf{0} \text{
		}\right]^T,~\bar x \in \mathbb{R}^2\}\subset X_2.
		\end{equation}
	\end{small}
	By merging \eqref{eq:X2_a} and \eqref{eq:X2_b}, we obtain
	\begin{small}
		\begin{equation*}
		\begin{aligned}
		X_2 =\big\{\mathbf{x} : \mathbf{x} =& \left[  \text{ }\bar{x}^T\text{ }\tilde x_{12}^T
		\text{ } \tilde x_{13}^T \text{ }|\text{ }\cdots \text{ }|\text{ }\bar{x}^T\text{ }
		\tilde x_{N2}^T \text{ } \tilde x_{N3}^T\text{ }\right]^T,\\
		& \bar{x}, \tilde{x}_{i2},\tilde{x}_{i3}\in\mathbb{R}^2\big\},
		\end{aligned}
		\end{equation*}
	\end{small}
	and then, from \eqref{eq:R}, it follows
	\begin{small}
		\begin{equation}
		\label{eqn:Rset}
		\begin{aligned}
		R = \big\{\mathbf{x} : \mathbf{x} =&\left[  \text{ }\bar\alpha^T\text{ }\beta_1^T
		\text{ } \gamma_1^T \text{ }|\text{ }\cdots \text{ }| \text{ }\bar\alpha^T\text{ }\beta_N^T
		\text{ } \gamma_N^T \text{ }\right]^T, \\
		&\bar\alpha,\beta_i,
		\gamma_i\in\mathbb{R}^2,
		\beta_i = {\bar\sigma}^{-1}\mathcal{Y}_{33,i}^{-1}\gamma_i\big\}.
		\end{aligned}
		\end{equation}
	\end{small}
	For concluding the proof, we must show that the largest
	invariant set $M\subseteq R$ is the origin. To this purpose, we
	consider \eqref{eq:modelDGUgen-aug-closed}, include coupling terms
	$\hat\xi_{[i]}$ and neglect inputs. Then,
	we choose the initial
	state $\mathbf{\hat x}(0) = \left[ \hat x_1^T(0)|\dots|\hat
	x^T_N(0)\right]^T\in R$, where, according to \eqref{eqn:Rset}, $\hat x_i(0)  =\left[  \text{ }\bar{\alpha}^T\text{ }\beta_i^T
	\text{ } \gamma_i^T \text{ }\right]^T$, $i = 1, \dots, N$. Our aim
	is to find conditions on the
	elements of  $\mathbf{\hat x}(0)$ that must hold in order to guarantee
	$\mathbf{\dot{\hat{x}}}\in R$. Recalling \eqref{eq:modelDGUgen-aug-closed} and
	\eqref{eq:Fi}, we compute $\dot{\hat x}_i(0)$ as
	
	\begin{equation}
	\label{eqn:dot_x_hat}
	\begin{aligned}
	&F_i\hat
	x_i(0)+\sum\limits_{j\in\mathcal{N}_i}\underbrace{\hat A_{ij}\left(\hat
		x_j(0)-\hat x_i(0)\right)}_{=0}=\\
	&=\left[\begin{array}{ccc}
	\mathcal{F}_{11,i} & \mathcal{F}_{12,i} & \mathbf{0_2} \\ 
	\mathcal{F}_{21,i}  &   \mathcal{F}_{22,i}  &  \mathcal{F}_{23,i} \\
	-I_2 & \mathbf{0_2} & \mathbf{0_2} \\
	\end{array}\right]\left[\begin{array}{c}
	\bar\alpha\\
	\beta_i \\
	\gamma_i
	\end{array}\right]=\left[\begin{array}{c}
	\mathcal{F}_{11,i}\bar\alpha +  \mathcal{F}_{12,i}\beta_i\\
	\mathcal{F}_{21,i}\bar\alpha + \mathcal{F}_{22,i}\beta_i + \mathcal{F}_{23,i}\gamma_i\\
	-\bar{\alpha}
	\end{array}\right],
	\end{aligned}
	\end{equation}
	\normalsize
	In order to have $\dot{\hat{x}}_i (0)\in R$, it must hold, $\forall i,j\in\mathcal{D}$
	\begin{equation}
	\label{eq:F11_F12}
	\mathcal{F}_{11,i}\bar\alpha +  \mathcal{F}_{12,i}\beta_i =
	\mathcal{F}_{11,j}\bar\alpha +  \mathcal{F}_{12,j}\beta_j
	\end{equation}
	Since, from \eqref{eqn:Fi_subplots}, $\mathcal{F}_{11,i}$ 
	is independent of $i$ and nonsingular, and since $\mathcal{F}_{12,\star} = \frac{1}{C_{t\star}}I_2$, $\star\in\{i,j\}$, \eqref{eq:F11_F12} is equivalent to 
	\begin{equation*}
	\frac{1}{C_{ti}}\beta_i =\frac{1}{C_{tj}}\beta_j.
	\end{equation*}
	This means that there is $\bar{\beta}\in\mathbb{R}^2$, such that $C_{ti}\bar{\beta} =
	\beta_i$, $\forall i\in\mathcal{D}$. Moreover, $\dot{\hat{x}}_i (0)\in R$ implies
	the following relation between the two last subvectors in \eqref{eqn:dot_x_hat}

	\begin{equation}
	\label{eq:subvec}
	\mathcal{F}_{21,i}\bar\alpha+ \mathcal{F}_{22,i}({\bar\sigma}^{-1}\YY_{33,i}^{-1}\gamma_i)+\FF_{23,i}\gamma_i=-{\bar\sigma}^{-1}\YY_{33,i}^{-1}\bar \alpha.
	\end{equation}
	\normalsize
	Our next aim is to show that the previous equation is verified only for $\bar \alpha=0$. To this purpose we perform the following intermediate computations.
	\begin{itemize}
		\item By setting $K_i=\matrice{ccc}{\KK_{11,i} & \KK_{12,i} & \KK_{13,i}}$, $\KK_{1j,i}\in\Rset^{2\times 2}$, $j=1,2,3$, 
		from $G_i=K_i Y_i$, \eqref{eq:Yi_param} and \eqref{eq:Y23} one has
		
		\begin{equation}
		\label{eq:Gblocks}
		\matrice{c|c|c}{\GG_{11,i} & \GG_{12,i} & \GG_{13,i}}=\matrice{c|c|c}{\eta_i^{-1}\KK_{11,i} & \KK_{12,i} \YY_{22,i}+{\bar\sigma}^{-1}\KK_{13,i} & {\bar\sigma}^{-1}\KK_{12,i}+\KK_{13,i} \YY_{33,i}}. \nonumber
		\end{equation} 
		\normalsize
		\item From the expression of $\GG_{11,i}$ in \eqref{eq:G11} and the relation $\GG_{11,i}=\eta_i^{-1}\KK_{11,i}$ obtained in \eqref{eq:Gblocks} one has
		\begin{equation}
		\label{eq:K21}
		\KK_{11,i}=\bar \sigma L_{ti} \YY_{22,i}+I_2.
		\end{equation}
		By direct computation, from \eqref{eq:modelDGUgen-aug-closed} and \eqref{eq:Fi} one has $\FF_{21,i}=L_{ti}^{-1}(-I_2+\KK_{11,i})$. Using \eqref{eq:K21} one obtains
		\begin{equation}
		\label{eq:F21}
		\FF_{21,i}=\bar\sigma\YY_{22,i}.
		\end{equation}
		\item From the expression of $\GG_{13,i}$ in \eqref{eq:G13} and \eqref{eq:Gblocks}, one has $\bar\sigma^{-1}\KK_{12,i}+\KK_{13,i}\YY_{33,i}=-L_{ti}\bar\sigma^{-1}\hAAA_{22,i}$, which implies
		\begin{equation}
		\label{eq:K31rel}
		\bar\sigma^{-1}(\KK_{12,i}+L_{ti}\hAAA_{22,i})=-\KK_{13,i}\YY_{33,i}.
		\end{equation}
	\end{itemize}	
	
	Substituting \eqref{eq:F21} in \eqref{eq:subvec} one obtains
	\begin{align}
	\label{eq:alpha}
	&(\bar \sigma \YY_{22,i}+\bar\sigma^{-1}\YY_{33,i}^{-1})\bar\alpha=\LL_i\gamma_i, \\
	\label{eq:LL}
	&\LL_i=-\FF_{22,i}\bar\sigma^{-1} \YY_{33,i} ^{-1} -\FF_{23,i}.
	\end{align}
	By direct computation, from \eqref{eq:modelDGUgen-aug-closed} and \eqref{eq:Fi}  one has $\FF_{22,i}=\hAAA_{22,i}+L_{ti}^{-1}\KK_{12,i}$ and $\FF_{23,i}=L_{ti}^{-1}\KK_{13,i}$. Hence,
	
	\begin{align}
	\label{eq:LLeq0}
	\LL_i=-L_{ti}^{-1}(L_{ti}\hAAA_{22,i}+\KK_{12,i})\bar\sigma^{-1} \YY_{33,i} ^{-1} -L_{ti}^{-1}\KK_{13,i}=0,
	\end{align}
	\normalsize
	where the last equality follows from \eqref{eq:K31rel}. In \eqref{eq:alpha}, the matrix $(\bar \sigma \YY_{22,i}+\bar\sigma^{-1}\YY_{33,i}^{-1})$ is nonsingular because $Y_i>0$ implies $\YY_{22,i}>0$ and $\YY_{33,i}>0$. Hence, \eqref{eq:alpha} is verified only by $\bar\alpha=0$, which is the desired result.\\
	As a consequence, in order to have $\dot{\hat{x}}_i (0)\in R$, it must
	hold $\hat{x}_i (0) = \left[ \text{ }\mathbf{0}\text{ }\beta_i^T
	\text{ } \gamma_i^T \text{ }\right]^T$, with
	$\beta_i=\bar\sigma^{-1}\YY_{33,i}^{-1}\gamma_i$.
	Let 
	\begin{equation}
	\label{eqn:setS}
	S = \big\{\mathbf{x}:\mathbf{x} =\left[  \text{ }\mathbf{0}\text{ }(\bar\sigma^{-1}\YY_{33,1}^{-1}\gamma_1)^T
	\text{ } \gamma_1^T\text{ }|\text{ }\cdots \text{ }| \text{ }\mathbf{0}\text{ }(\bar\sigma^{-1}\YY_{33,N}^{-1}\gamma_N)^T\text{ } \gamma_N^T\right]
	\big\}.
	\end{equation}
	\normalsize
	Since $M\subseteq S$, we pick $\mathbf{\tilde{x}}(0)\in S$ and impose
	$\mathbf{\dot{\tilde{x}}}(0)\in S$. 
	Using \eqref{eqn:Fi_subplots}, this gives
	\begin{equation*}
	\dot{\tilde{x}}_i(0) = F_i\left[\begin{array}{c}
	\mathbf{0} \\
	\bar\sigma^{-1}\YY_{33,i}^{-1}\gamma_i\\
	\gamma_i
	\end{array}\right]=\left[\begin{array}{c}
	\bar\sigma^{-1}C_{ti}^{-1}\YY_{33,i}^{-1}\gamma_i \\
	-\LL_i\gamma_i\\
	\mathbf{0}
	\end{array}\right], \forall i \in\mathcal D.
	\end{equation*}
	Since, from \eqref{eq:LLeq0}, $\LL_i=0$, one has that $\dot{\tilde{x}}_i(0)\in S$ if and only if $\bar\sigma^{-1}C_{ti}^{-1}\YY_{33,i}^{-1}\gamma_i=0$, which is verified only for $\gamma_i=0$.
	
	From \eqref{eqn:setS} one has $S=\{\mathbf{0}\}$ and, since $M\subseteq S$, it holds $M =\{\mathbf{0}\}$. 
       \end{proof}
       \clearpage
  
  	   \bibliographystyle{IEEEtran}
  	\bibliography{PnP_AC_line_independent}

\begin{thebibliography}{10}
\providecommand{\url}[1]{#1}
\csname url@samestyle\endcsname
\providecommand{\newblock}{\relax}
\providecommand{\bibinfo}[2]{#2}
\providecommand{\BIBentrySTDinterwordspacing}{\spaceskip=0pt\relax}
\providecommand{\BIBentryALTinterwordstretchfactor}{4}
\providecommand{\BIBentryALTinterwordspacing}{\spaceskip=\fontdimen2\font plus
\BIBentryALTinterwordstretchfactor\fontdimen3\font minus
  \fontdimen4\font\relax}
\providecommand{\BIBforeignlanguage}[2]{{%
\expandafter\ifx\csname l@#1\endcsname\relax
\typeout{** WARNING: IEEEtran.bst: No hyphenation pattern has been}%
\typeout{** loaded for the language `#1'. Using the pattern for}%
\typeout{** the default language instead.}%
\else
\language=\csname l@#1\endcsname
\fi
#2}}
\providecommand{\BIBdecl}{\relax}
\BIBdecl

\bibitem{Guerrero2013}
J.~M. Guerrero, M.~Chandorkar, T.~Lee, and P.~C. Loh, ``{Advanced Control
  Architectures for Intelligent Microgrids — Part I: Decentralized and
  Hierarchical Control},'' \emph{IEEE Transactions on Industrial Electronics},
  vol.~60, no.~4, pp. 1254--1262, 2013.

\bibitem{schiffer2014conditions}
J.~Schiffer, R.~Ortega, A.~Astolfi, J.~Raisch, and T.~Sezi, ``Conditions for
  stability of droop-controlled inverter-based microgrids,'' \emph{Automatica},
  vol.~50, no.~10, pp. 2457--2469, 2014.

\bibitem{Simpson-Porco2013}
J.~W. Simpson-Porco, F.~D\"{o}rfler, and F.~Bullo, ``{Synchronization and power
  sharing for droop-controlled inverters in islanded microgrids},''
  \emph{Automatica}, vol.~49, no.~9, pp. 2603--2611, 2013.

\bibitem{7172122}
M.~Andreasson, R.~Wiget, D.~V. Dimarogonas, K.~H. Johansson, and G.~Andersson,
  ``Distributed primary frequency control through multi-terminal hvdc
  transmission systems,'' in \emph{2015 American Control Conference (ACC)},
  July 2015, pp. 5029--5034.

\bibitem{Etemadi2012a}
A.~H. Etemadi, E.~J. Davison, and R.~Iravani, ``{A Decentralized Robust Control
  Strategy for Multi-DER Microgrids - Part I: Fundamental Concepts},''
  \emph{IEEE Transactions on Power Delivery}, vol.~27, no.~4, pp. 1843--1853,
  2012.

\bibitem{Babazadeh2013}
M.~Babazadeh and H.~Karimi, ``{A Robust Two-Degree-of-Freedom Control Strategy
  for an Islanded Microgrid},'' \emph{IEEE Transactions on Power Delivery},
  vol.~28, no.~3, pp. 1339--1347, 2013.

\bibitem{7062912}
V.~Nasirian, Q.~Shafiee, J.~M. Guerrero, F.~L. Lewis, and A.~Davoudi,
  ``Droop-free distributed control for ac microgrids,'' \emph{IEEE Transactions
  on Power Electronics}, vol.~31, no.~2, pp. 1600--1617, Feb 2016.

\bibitem{Riverso_TSG}
S.~Riverso, F.~Sarzo, and G.~Ferrari-Trecate, ``Plug-and-play voltage and
  frequency control of islanded microgrids with meshed topology,'' \emph{IEEE
  Transactions on Smart Grid}, vol.~6, no.~3, pp. 1176--1184, May 2015.

\bibitem{IEEE2017}
IEEE1588, ``{IEEE} standard profile for use of {IEEE} 1588 precision time
  protocol in power system applications,'' \emph{IEEE C37.238-2017, August
  2017}, pp. 1--42, Aug 2017.

\bibitem{kumagai2014rise}
J.~Kumagai, ``The rise of the personal power plant,'' \emph{IEEE Spectrum, The
  smarter grid}, vol.~51, no.~6, pp. 54--49, 2014.

\bibitem{7511679}
M.~S. Sadabadi, Q.~Shafiee, and A.~Karimi, ``Plug-and-play voltage
  stabilization in inverter-interfaced microgrids via a robust control
  strategy,'' \emph{IEEE Transactions on Control Systems Technology}, vol.~PP,
  no.~99, pp. 1--11, 2016.

\bibitem{tucci2015decentralized}
M.~Tucci, S.~Riverso, J.~C. Vasquez, J.~M. Guerrero, and G.~Ferrari-Trecate,
  ``A decentralized scalable approach to voltage control of {DC} islanded
  microgrids,'' \emph{IEEE Transactions on Control Systems Technology},
  vol.~24, no.~6, pp. 1965--1979, Nov 2016.

\bibitem{7039383}
F.~D\"{o}rfler, J.~W. Simpson-Porco, and F.~Bullo, ``Plug-and-play control and
  optimization in microgrids,'' in \emph{53rd IEEE Conference on Decision and
  Control}, Dec 2014, pp. 211--216.

\bibitem{7040312}
S.~Bansal, M.~N. Zeilinger, and C.~J. Tomlin, ``Plug-and-play model predictive
  control for electric vehicle charging and voltage control in smart grids,''
  in \emph{53rd IEEE Conference on Decision and Control}, Dec 2014, pp.
  5894--5900.

\bibitem{lucia2015contract}
S.~Lucia, M.~K{\"o}gel, and R.~Findeisen, ``Contract-based predictive control
  of distributed systems with plug and play capabilities,''
  \emph{IFAC-PapersOnLine}, vol.~48, no.~23, pp. 205--211, 2015.

\bibitem{bodenburg2015plug}
S.~Bodenburg and J.~Lunze, ``Plug-and-play control of interconnected systems
  with a changing number of subsystems,'' in \emph{Control Conference (ECC),
  2015 European}.\hskip 1em plus 0.5em minus 0.4em\relax IEEE, 2015, pp.
  3520--3527.

\bibitem{bendtsen2013plug}
J.~Bendtsen, K.~Trangbaek, and J.~Stoustrup, ``Plug-and-play
  control—modifying control systems online,'' \emph{IEEE Transactions on
  Control Systems Technology}, vol.~21, no.~1, pp. 79--93, 2013.

\bibitem{Riverso2013c}
S.~Riverso, M.~Farina, and G.~Ferrari-Trecate, ``{Plug-and-Play Decentralized
  Model Predictive Control for Linear Systems},'' \emph{IEEE Transactions on
  Automatic Control}, vol.~58, no.~10, pp. 2608--2614, 2013.

\bibitem{Tucci2016independent}
\BIBentryALTinterwordspacing
M.~Tucci, S.~Riverso, and G.~Ferrari-Trecate, ``{Voltage stabilization in {DC}
  microgrids: an approach based on line-independent plug-and-play
  controllers},'' Pavia, Italy, Tech. Rep., 2016. [Online]. Available:
  \url{arXiv:1609.02456}
\BIBentrySTDinterwordspacing

\bibitem{schiffer2016survey}
J.~Schiffer, D.~Zonetti, R.~Ortega, A.~M. Stankovi{\'c}, T.~Sezi, and
  J.~Raisch, ``{A survey on modeling of microgrids - From fundamental physics
  to phasors and voltage sources},'' \emph{Automatica}, vol.~74, pp. 135--150,
  2016.

\bibitem{dorfler2013kron}
F.~D{\"o}rfler and F.~Bullo, ``Kron reduction of graphs with applications to
  electrical networks,'' \emph{IEEE Transactions on Circuits and Systems I:
  Regular Papers}, vol.~60, no.~1, pp. 150--163, 2013.

\bibitem{TucciFloriduz_Kron}
M.~Tucci, A.~Floriduz, S.~Riverso, and G.~Ferrari-Trecate, ``Plug-and-play
  control of {AC} islanded microgrids with general topology,'' in
  \emph{European Control Conference (ECC)}, 2015, pp. 1493--1500.

\bibitem{Riverso2014c}
\BIBentryALTinterwordspacing
S.~Riverso, F.~Sarzo, and G.~Ferrari-Trecate, ``{Plug-and-play voltage and
  frequency control of islanded microgrids with meshed topology},'' Pavia,
  Italy, Tech. Rep., 2014. [Online]. Available: \url{arXiv:1405.2421}
\BIBentrySTDinterwordspacing

\bibitem{godsil2001algebraic}
C.~Godsil and G.~Royle, ``Algebraic graph theory, volume 207 of {G}raduate
  {T}exts in {M}athematics,'' 2001.

\bibitem{Boyd1994}
S.~Boyd, L.~{El Ghaoui}, E.~Feron, and V.~Balakrishnan, \emph{{Linear matrix
  inequalities in system and control theory}}.\hskip 1em plus 0.5em minus
  0.4em\relax Philadelphia, Pennsylvania, USA: SIAM Studies in Applied
  Mathematics, vol. 15, 1994.

\bibitem{IEEE2009}
IEEE, ``{IEEE Recommended Practice for Monitoring Electric Power Quality},''
  IEEE Std 1159™-2009, 3 Park Avenue, New York, NY 10016-5997, USA, 29 June,
  Tech. Rep. June, 2009.

\bibitem{4579760}
``{IEEE} standard for a precision clock synchronization protocol for networked
  measurement and control systems,'' \emph{{IEEE Std 1588-2008 (Revision of
  IEEE Std 1588-2002)}}, pp. 1--300, July 2008.

\bibitem{gusella1989accuracy}
R.~Gusella and S.~Zatt, ``The accuracy of the clock synchronization achieved by
  tempo in berkeley unix 4.3 bsd,'' \emph{IEEE Trans. Softw. Eng.}, vol.~15,
  no.~7, pp. 847--853, 1989.

\bibitem{khalil2001nonlinear}
H.~K. Khalil, \emph{Nonlinear systems (3rd edition)}.\hskip 1em plus 0.5em
  minus 0.4em\relax Prentice Hall, 2001.

\end{thebibliography}
  \end{document}